\documentclass[11pt]{article}
\usepackage{amsmath,amssymb,amsfonts,amsthm,epsfig}
\usepackage[usenames,dvipsnames]{xcolor}
\usepackage{bm,xspace}
\usepackage{tcolorbox}
\usepackage{cancel}
\usepackage{fullpage}
\usepackage{liyang}
\usepackage{framed}
\usepackage{verbatim}
\usepackage{enumitem}
\usepackage{array}
\usepackage{multirow}
\usepackage{afterpage}
\usepackage{mathrsfs}
\usepackage{pifont} 
\usepackage{chngpage}
\usepackage[normalem]{ulem}
\usepackage{boxedminipage}
\usepackage{caption}
\usepackage{subcaption}

\usepackage{forest}

\usepackage{algorithm}
\usepackage{algorithmicx}
\usepackage{algpseudocode}

\usepackage{tikz}
\usetikzlibrary{calc,through,backgrounds,decorations.pathreplacing, calligraphy,arrows.meta}
\usetikzlibrary{positioning,chains,fit,shapes}

\usepackage{tikz-cd}

\usepackage{pgfplots}
\pgfplotsset{width=8cm,compat=newest}

\def\colorful{1}

\ifnum\colorful=1
\newcommand{\violet}[1]{{\color{violet}{#1}}}

\newcommand{\teal}[1]{{\color{teal}{#1}}}

\fi
\ifnum\colorful=0
\newcommand{\violet}[1]{{{#1}}}

\fi

\newcommand{\depth}{\mathrm{depth}}

\def\setcover{\textsc{Set-Cover}~}

\def\width{\mathrm{width}}

\newcommand{\pryd}[1]{\underset{\by\sim \mathcal{D}_{\oplus \ell}}{\Pr}\left[{#1}\right]}
\newcommand{\eyd}[1]{\underset{\by\sim \mathcal{D}_{\oplus \ell}}{\E}\left[{#1}\right]}

\newcommand{\prydj}[1]{\underset{\by\sim \mathcal{D}^j_{\oplus \ell}}{\Pr}\left[{#1}\right]}

\newcommand{\prxd}[1]{\underset{\bx\sim \mathcal{D}}{\Pr}\left[{#1}\right]}

\newcommand{\exd}[1]{\underset{\bx\sim \mathcal{D}}{\E}\left[{#1}\right]}

\newtheorem{hypothesis}{Hypothesis}

\newcommand{\paren}[1]{\left({#1}\right)}

\newlist{enumprop}{enumerate}{1} 
\setlist[enumprop]{label=\arabic*.,ref=\theproposition.\arabic*}
\crefalias{enumpropi}{proposition}

\makeatletter
\newtheorem*{rep@theorem}{\rep@title}
\newcommand{\newreptheorem}[2]{
\newenvironment{rep#1}[1]{
 \def\rep@title{#2 \ref{##1}}
 \begin{rep@theorem}\itshape}
 {\end{rep@theorem}}}
\makeatother

\newreptheorem{theorem}{Theorem}

\begin{document}


\title{Superpolynomial lower bounds for decision tree learning and testing \vspace{1pt}}

\author{ 
Caleb Koch \vspace{6pt} \\ 
{\small {\sl Stanford}} \and 
\hspace{5pt} Carmen Strassle \vspace{6pt} \\
\hspace{5pt} {\small {\sl Stanford}} \and 
Li-Yang Tan \vspace{6pt}  \\
\hspace{-10pt} {\small {\sl Stanford}}
}

\date{\vspace{15pt}\small{\today}}

 \maketitle

\begin{abstract}



{
We establish new hardness results for decision tree optimization problems, adding to a line of work that dates back to Hyafil and Rivest in 1976. We prove, under the randomized exponential time hypothesis,  superpolynomial runtime lower bounds for two basic problems: given an explicit representation of a function $f$ and a generator for a distribution $\mathcal{D}$, 
 \begin{itemize} 
 \item[$\circ$]  {\sl construct} a small decision tree approximator for $f$ under $\mathcal{D}$, and 
 \item[$\circ$]  {\sl decide} if there is a small decision tree approximator for $f$ under $\mathcal{D}$.  
 \end{itemize} 
Our results imply new lower bounds for distribution-free PAC learning and testing of decision trees, settings in which the algorithm only has restricted access to $f$ and $\mathcal{D}$.
  Specifically, we get that: 
\begin{itemize}
\item[$\circ$] $n$-variable size-$s$ decision trees cannot be properly PAC learned in time $n^{\tilde{O}(\log\log s)}$, and 
\item[$\circ$] depth-$d$ decision trees cannot be tested in time $\exp(d^{\,O(1)})$.
\end{itemize}
For learning, the previous best lower bound only ruled out $\poly(n)$-time algorithms (Alekhnovich, Braverman, Feldman, Klivans, and Pitassi, 2009).  For testing, recent work gives  similar though incomparable lower bounds in the setting where $f$ is random and $\mathcal{D}$ is nonexplicit (Blais, Ferreira Pinto Jr., and Harms, 2021).

Assuming a plausible conjecture on the hardness of {\sc Set-Cover}, we show that our lower bound for properly PAC learning decision trees can be improved to $n^{\Omega(\log s)}$, matching the best known upper bound of $n^{O(\log s)}$ due to Ehrenfeucht and Haussler (1989).

We obtain our results within a unified framework that leverages recent progress in two different lines of work: the inapproximability of {\sc Set-Cover} and XOR lemmas for query complexity. Our framework is versatile and yields results for related concept classes such as juntas and DNF formulas. 
}

\end{abstract} 

\thispagestyle{empty}
\newpage 

\setcounter{page}{1}
\section{Introduction}

The algorithmic problem of constructing decision tree representations of functions is one of the most basic and well-studied problems of computer science. Greedy decision tree learning heuristics such as ID3, C4.5, and CART, developed in the 1980s, continue to be indispensable to everyday machine learning and enjoy empirical success.  The data mining textbook~\cite{WFHP16} describes C4.5 as ``a landmark decision tree program that is probably the machine learning workhorse most widely used in practice to date".   In addition to being extremely fast to evaluate, a key advantage of decision trees is their simple and easy-to-understand structure, making them the most canonical example of an interpretable model. The recent survey~\cite{Rud22} lists decision tree learning as the very first of ``10 grand challenges" in the emerging field of interpretable machine learning. 

In terms of algorithms with theoretical guarantees, a classic result of Ehrenfeucht and Haussler~\cite{EH89} gives a quasipolynomial time algorithm for properly PAC learning decision trees: Given  labeled examples $(\bx,f(\bx))$ where $f : \zo^n\to\zo$ can be computed by a size-$s$ decision tree and $\bx$ is drawn from a distribution $\mathcal{D}$ over $\zo^n$, their algorithm runs in $n^{O(\log s)}$ time and returns a decision tree hypothesis that is close to $f$ under $\mathcal{D}$. Numerous alternative algorithms have since been designed within restricted variants of the PAC model (e.g.~where $\mathcal{D}$ is assumed to be uniform) and by relaxing the problem (e.g.~allowing hypotheses that are not themselves decision trees\footnote{Such improper decision tree learning algorithms do not apply to the problem of ``decision tree learning" as is meant in the context of machine learning, where it always refers to the problem of constructing decision tree {\sl hypotheses}. See e.g.~the textbooks~\cite{Mit97,Bis06,SS14} or the Wikipedia page for ``Decision tree learning".  From a practical perspective, properness of decision tree algorithms is not just a feature but the entire point---to produce a decision tree representation of the data.  The focus of this paper will be on proper decision tree learning algorithms.})~\cite{Riv87,Blu92,Han93,KM93,KM96,Bsh93,GLR99,BM02,MR02,JS05,KS06,OS07,GKK08,KST09,HKY18,CM19,BLQT21focs}, but Ehrenfeucht and Haussler's algorithm remains state of the art in the standard PAC model. 

Another interesting setting is when an explicit representation of the function $f$, and possibly also the distribution $\mathcal{D}$, are given to the algorithm.  This easier setting, where the algorithm can ``inspect" $f$, models a popular approach in explainable machine learning known as post-hoc explanations.  The goal here is not to train a decision tree model for an unknown function $f$, but instead to turn a complicated trained model~$f$ (e.g.~a neural net) into its decision tree representation. While numerous algorithms for this task have been proposed in the empirical literature~\cite{CS95,BS96,AB07,ZH16,BKB17,VLJODV17,FH17,VS20}, among those with theoretical guarantees, the fastest one remains that of Ehrenfeucht and Haussler.

In parallel with these lines of algorithmic work, there has also been a similarly large body of work on the {\sl hardness} of decision tree learning~\cite{HR76,GJ79,BFJKMR94,HJLT96,KPB99,ZB00,LN04,CPRAM07,RRV07,Sie08,ABFKP09,AH12,Rav13,BLT-ITCS}.  It is interesting to note that the earliest paper here, by Hyafil and Rivest in 1976, predates Ehrenfeucht and Haussler's algorithm by more than a decade; indeed, it even predates the PAC model.  Their paper, which established the NP-completeness of a certain formulation of decision tree learning with perfect accuracy, reveals that the problem was already intensively studied and recognized as central in the 1970s. Quoting the authors,  ``the importance of this result can be measured in terms of the large amount of effort that has been put into finding efficient algorithms for constructing optimal binary decision trees".  

A closely related problem is that of {\sl testing} decision trees: while in learning one is interested in constructing small decision trees, here the goal is simply to  {\sl decide} if one such tree exists. The distribution-free model of property testing, introduced by Goldreich, Goldwasser, and Ron~\cite{GGR98} to parallel distribution-free PAC learning, has received increasing attention in recent years~\cite{CX16,LCSSX18,Bsh19,Bel19,Har19,FY20,RR20,BPH21,Bsh22,BHZ22,ABFKY22,CP22,HY22}. 

\section{Our results} 
We establish new hardness results for distribution-free learning and testing of decision trees. 
For both problems, our lower bounds hold even when explicit representations of both the  function $f$ and distribution $\mathcal{D}$ are given to the algorithm; lower bounds in this setting imply lower bounds for learning and testing.   

We obtain our results within a unified framework that brings together two active lines of research: the inapproximability of {\sc Set-Cover}~\cite{LY94,Fei98, CHKX06,DS14,Mos15,KLM18,CL19,Lin19,CHK20,KI21} 
and XOR lemmas for query complexity~\cite{Dru12,BB19,BPH21}.  Connections between {\sc Set-Cover} and decision tree optimization problems, both in terms of algorithms and hardness, date back to~\cite{HR76} and are present in numerous prior works; we leverage recent progress in both the parameterized and nonparameterized settings.  The connection to XOR lemmas, on the other hand, is new to this work.  All our lower bounds, being computational in nature, are conditioned on the randomized Exponential Time Hypothesis (ETH).  As a byproduct, our lower bounds hold even against randomized algorithms.

We now give a detailed overview of our results, in tandem with a discussion of how they compare with prior work.  

\subsection{Lower bounds for {\sc DT-Construction}}

The {\sc DT-Construction} problem is the variant of decision tree learning where $f$ and $\mathcal{D}$ are both given to the algorithm:
\medskip 

\begin{tcolorbox}[colback = white,arc=1mm, boxrule=0.25mm]
\vspace{3pt} 
{\sc DT-Construction$(s,\eps)$:}  Given as input a circuit representation of a function $f : \zo^n \to \zo$, a generator for a distribution $\mathcal{D}$ over $\zo^n$, parameters $s\in \N$ and $\eps \in (0,1)$, and the promise that $f$ is a size-$s$ decision tree under $\mathcal{D}$, construct a decision tree $T$ that is $\eps$-close to $f$ under $\mathcal{D}$.
\vspace{3pt} 
\end{tcolorbox}
\medskip

Our first result is a superpolynomial runtime lower bound for {\sc DT Construction}:
\begin{theorem}
\label{thm:DT-DT} 
Under randomized ETH, for $s = n$ and $\eps = \frac1{n}$ any algorithm for {\sc DT-Construction}$(s,\eps)$ must take $n^{\tilde{\Omega}(\log\log s)}$ time. 
\end{theorem} 

Prior works also focused on the parameter settings $s=n$ and $\eps = \frac1{n}$, corresponding to strong learning of linear-size decision trees.  Most recently, Alekhnovich, Braverman, Feldman, Klivans, and Pitassi~\cite{ABFKP09} ruled out $\poly(n)$ time algorithms under the assumption that ${\textsc{Sat}}$ cannot be solved in randomized subexponential time. Before that, Hancock, Jiang, Li, and Tromp~\cite{HJLT96} ruled out $\poly(n)$ time algorithms that return a decision tree hypothesis of size $n^{1+o(1)}$, under the assumption that $\textsc{Sat}$ cannot be solved in randomized quasipolynomial time.

Our proof of~\Cref{thm:DT-DT} opens up a concrete route towards obtaining the optimal $n^{\Omega(\log s)}$ lower bound.  We can also show an $n^{\Omega(\log s)}$ lower bound for the stricter version of {\sc DT-Construction} where the algorithm has to return a decision tree of size $s$ (instead of one of any size).  We elaborate on both of these in~\Cref{sec:optimal}.


\paragraph{Hardness of learning juntas with DNF hypotheses.} We obtain~\Cref{thm:DT-DT} as a corollary of our first main result, which simultaneously allows for a stronger promise on the simplicity of the target function $f$ and for the algorithm to return a more expressive hypothesis:  

\begin{theorem} 
\label{thm:junta-DNF}
Under randomized ETH, for $s = n$ and $\eps = \frac1{n}$ any algorithm for {\sc DT-Construction}$(s,\eps)$ must take $n^{\tilde{\Omega}(\log\log s)}$ time, even if $f$ is further promised to be a {\sl $(\log s)$-junta} under $\mathcal{D}$ and the algorithm is allowed to return a {\sl DNF hypothesis}. 
\end{theorem} 

We recall the strict inclusions
\[ \{ \text{$(\log s)$-juntas} \} \subset \{ \text{size-$s$ decision trees}\} \subset \{ \text{size-$s$ DNFs} \}. \] 
Each class is exponentially more expressive than the previous one: a size-$s$ decision tree can depend on as many as $s$ variables
, and a size-$s$ DNF can require a decision tree of size $2^{\Omega(s)}$.

The results of~\cite{ABFKP09,HJLT96} are not known to be amenable to such a strengthening. \cite{ABFKP09} did give lower bounds for {\sc DNF-Construction}, the analogue of {\sc DT-Construction} where the target $f$ is promised to be a DNF under $\mathcal{D}$ and the algorithm is expected to construct a DNF hypothesis.  They ruled out $\poly(n)$ time algorithms for $s=n$ and $\eps = \frac1{n}$.  
\cite{ABFKP09} gave two separate proofs of hardness for {\sc DT-Construction} and {\sc DNF-Construction}, reducing from {\sc Set-Cover} for the former and from {\sc Chromatic-Number} for the latter.  \Cref{thm:junta-DNF}, on the other hand, yields new lower bounds for both problems via a single proof.  

\paragraph{Hardness of properly learning  juntas.}  
Implicit in the proofs of~\Cref{thm:DT-DT,thm:junta-DNF} is a tight connection between algorithms for {\sc Set-Cover} and algorithms for properly learning juntas.  By making this connection explicit, we obtain strong lower bounds for the latter problem that hold even under the promise that the target is a {\sl monotone disjunction}: 

\begin{theorem}
\label{thm:junta-junta} 
Under randomized ETH, for any $k \le n^c$ where $c<1$ is any constant and $\eps = O(\frac1{n})$, there is no algorithm that, given as input a circuit representation of a function $f : \zo^n \to \zo$, a generator for a distribution $\mathcal{D}$  and the promise that $f$ is a monotone $k$-disjunction under $\mathcal{D}$, runs in $n^{o(k)}$ time and constructs a $k$-junta that is $\eps$-close to $f$ under $\mathcal{D}$.  Under randomized SETH, we get a lower bound of $O(n^{k-\lambda})$ for any constant $\lambda > 0$. \end{theorem} 

These lower bounds nearly match the $O(n^k/\eps)$ runtime algorithm of the trivial algorithm that iterates over all possible $k$-junta hypotheses.  Previously,~\cite{ABFKP09} ruled out $\poly(n)$-time algorithms for $k\le O(\log n)$.

\subsection{Lower bounds for {\sc DT-Estimation}} 

The second problem that we consider, {\sc DT-Estimation}, is a variant of distribution-free decision tree testing where $f$ and $\mathcal{D}$ are both given to the algorithm:\footnote{It will be more convenient for us to measure the complexity of decision trees by their depth in this section, though there are direct analogues of our results for size instead of depth.} \medskip

\begin{tcolorbox}[colback = white,arc=1mm, boxrule=0.25mm]
\vspace{3pt} 
 {\sc DT-Estimation$(d,\eps)$}:  Given as input a circuit representation of a function $f : \zo^n \to \zo$, a generator for a distribution $\mathcal{D}$ over $\zo^n$, and parameters $d\in \N,\eps \in (0,1)$, distinguish between the following cases:
 \begin{itemize} 
 \item[$\circ$] {\sc Yes}: $f$ is a depth-$d$ decision tree under $\mathcal{D}$. 
 \item[$\circ$] {\sc No}: $f$ is $\eps$-far from every depth-$d$ decision tree under $\mathcal{D}$. 
 \end{itemize} 
\end{tcolorbox}
\medskip


Our second main result is an exponential lower bound for {\sc DT-Estimation}: 

\begin{theorem} 
\label{thm:test}
Under randomized ETH, any algorithm for {\sc DT-Estimation}$(d,\eps)$ must take $\exp(d^{\,\Omega(1)})$ time.  This holds even if $\eps = \frac1{2}-\exp(-d^{\,\Omega(1)})$ and the {\sc No} case satisfies the stronger promise that $f$ is $\eps$-far from every decision tree of depth $\Omega(d\log d)$ under $\mathcal{D}$. 
\end{theorem} 

Recent work of Blais, Ferreira Pinto Jr., and Harms~\cite{BPH21} gives 
an $\tilde{\Omega}(2^d)$ lower bound on the {\sl query complexity} testing of depth-$d$ decision trees.  This lower bound, however, only applies in the setting where both $f$ and $\mathcal{D}$ are unknown to the algorithm, since it is based on a random function $f$ and a nonexplicit  distribution $\mathcal{D}$.\footnote{This nonexplicit distribution is derived from lower bounds on the sample complexity of estimating distribution support size~\cite{WY19}.}  In contrast, our proof of~\Cref{thm:test} is constructive: it is based on an $f$ that is a depth-$3$ circuit (with $\{ \oplus, \vee\}$ gates) and a similarly simple generator for $\mathcal{D}$. Furthermore, \cite{BPH21}'s lower bound only holds when $\eps$ is a sufficiently small constant, whereas ours holds for~$\eps$ being exponentially close to $\frac1{2}$, and with a gap between the decision tree depths of the {\sc Yes} and {\sc No} cases.

As for upper bounds, Bshouty and Haddad-Zaknoon~\cite{BHZ22} give a distribution-free tester that runs in $2^{O(d)}n$ time and distinguishes depth-$d$ decision trees from those that are $\eps$-far from decision trees of depth $d^2$.  Under the uniform distribution, Blanc, Lange, and Tan~\cite{BLT-ICALP22} give an algorithm that runs in $\poly(d,1/\eps)\cdot n\log n$ time and distinguishes depth-$d$ decision trees from those that are $\eps$-far from decision trees of depth $O(d^3/\eps^3)$.

\subsection{Towards stronger lower bounds for {\sc DT-Construction}} 
\label{sec:optimal}

We show two ways in which the lower bounds of~\Cref{thm:DT-DT,thm:junta-DNF} can be further improved to $n^{\Omega(\log s)}$.  First,  we consider the stricter version of {\sc DT-Construction} where the algorithm has to return a size-$s$ decision tree:  
\begin{sloppypar}
\begin{theorem}
\label{thm:proper} Under randomized ETH, for $s = \exp(\tilde{O}(\log\log n))$ and $\eps = \frac1{n}$ any algorithm for {\sc DT-Construction}$(s,\eps)$ must take $n^{\Omega(\log s)}$ time if the algorithm has to return a size-$s$ decision tree.   As in~\Cref{thm:junta-DNF}, this holds even if $f$ is further promised to be a  $(\log s)$-junta under $\mathcal{D}$ and the algorithm is allowed to return a size-$s$ DNF hypothesis. 
\end{theorem} 
\end{sloppypar}
This more stringent version of {\sc DT-Construction} corresponds to the notion of {\sl strictly proper} learning of size-$s$ decision trees, where the algorithm has to return a hypothesis that falls within the concept class. Ehrenfeucht and Haussler's algorithm is not strictly proper.  On the other hand, for size-$s$ decision trees of depth $O(\log s)$, there is a simple dynamic programming  algorithm that runs in $n^{O(\log s)}$ time and is strictly proper~\cite{GLR99,MR02}.  Since every $(\log s)$-junta is a decision tree of depth $\log s$, this matches the lower bound of~\Cref{thm:proper}.

Finally, we show how an optimal lower bound of $n^{\Omega(\log s)}$ for the original version of {\sc DT-Construction}, matching the runtime of Ehrenfeucht and Haussler's algorithm, would follow from a natural and well-studied conjecture about {\sc Set-Cover}:  

\begin{conjecture}[Optimal inapproximability of parameterized {\sc Set-Cover}]
\label{conj:optimal-SetCover-hardness} 
There exists constants $\alpha,\beta < 1$ such that for $k\le N^{\alpha}$, there is no $N^{o(k)}$ time algorithm that, given a size-$N$ set cover instance, distinguishes between: 
\begin{itemize} 
\item[$\circ$] {\sc Yes}: There is a set cover of size $k$. 
\item[$\circ$] {\sc No}: Every set cover has size at least $k\cdot (1-\beta)\ln N$. 
\end{itemize} 
\end{conjecture}

There is a simple and efficient $\ln N$-approximation algorithm for {\sc Set-Cover}, and various hardness results are known for the problem of achieving a better approximation ratio~\cite{LY94,Fei98,DS14,Mos15,CHK20}.  \Cref{conj:optimal-SetCover-hardness} states that this hardness carries over to the {\sl parameterized} setting. Existing ETH-based lower bounds for parameterized {\sc Set-Cover}~\cite{CHKX06,KLM18,CL19,Lin19,KI21} are evidence in favor of it, and it is plausible that~\Cref{conj:optimal-SetCover-hardness} can be shown to hold under ETH.\footnote{See~\cite{MPW19,GKMP20} for further discussions of this conjecture and its implications for proof complexity.} We show:

\begin{theorem} 
\label{thm:opt} 
Under~\Cref{conj:optimal-SetCover-hardness}, for $s = n$ and $\eps = \frac1{n}$ any algorithm for {\sc DT-Construction}$(s,\eps)$ must take $n^{\Omega(\log s)}$ time.   As in~\Cref{thm:junta-DNF}, this holds even if $f$ is further promised to be a  $(\log s)$-junta under $\mathcal{D}$ and the algorithm is allowed to return a DNF hypothesis. 
\end{theorem} 

\Cref{table} summarizes our results for {\sc DT-Construction} and shows how they compare with the prior state of the art. 

\begin{table}[H]
\caption{Algorithms and lower bounds for {\sc DT-Construction}.  All our results are conditioned on randomized ETH, the lower bounds of~\Cref{thm:proper,thm:opt} are optimal.}  
\label{table} 
\renewcommand{\arraystretch}{2}
\centering
\begin{tabular}{|c|c|c|c|}
\hline
{\bf Reference} & {\bf Target} & ~~{\bf Hypothesis}~~ & {\bf Time complexity} \\
\hline 
\hline 
~~\cite{ABFKP09}~~ & size-$s$ DT & DT & ~~\text{$n^{\omega(1)}$ lower bound}~~  \\ \hline 
\cite{ABFKP09} & size-$s$ DNF & DNF & \text{$n^{\omega(1)}$ lower bound} \\ \hline 
\cite{EH89} & ~~size-$s$ DT~~ & DT & \text{~~$n^{O(\log s)}$ upper bound~~} \\ \hline   
~~~~~\Cref{thm:junta-DNF}~~~~~ & ~~$(\log s)$-junta~~ & DNF & ~~~\text{$n^{\tilde{\Omega}(\log\log s)}$ lower bound}~~~ \\ \hline 
\Cref{thm:proper} & ~~$(\log s)$-junta~~ & size-$s$ DNF & $n^{\Omega(\log s)}$ lower bound \\ \hline 
\Cref{thm:opt} & ~~$(\log s)$-junta~~ & DNF &  \begin{tabular}{@{}c@{}}
$n^{\Omega(\log s)}$ lower bound \vspace{-12pt} \\
under~\Cref{conj:optimal-SetCover-hardness} \end{tabular} \\ \hline
\end{tabular} 
\end{table}

\section{Our techniques} 
\label{sec:techniques} 

The starting point of all our reductions is the  parameterized version of {\sc Set-Cover}.  For a set cover instance $\mathcal{S}$, we write $\opt(\mathcal{S})$ to denote the size of the smallest set cover. 

\begin{definition} 
\label{def:setcover} 
The $(k,k')$-{\sc Set-Cover} problem is the following. Given as input a set cover instance $\mathcal{S}$ and parameters $k,k'\in \N$, output {\sc Yes} if $\opt(\mathcal{S}) \le k$ and {\sc No} if $\opt(\mathcal{S}) > k'$. 
\end{definition}

\paragraph{Reducing from {\sc Set-Cover} to juntas vs.~DNFs.} Our key lemma, which is the crux of our lower bounds for both {\sc DT-Construction} and {\sc DT-Estimation}, is a reduction from $(k,k')$-{\sc Set-Cover} to the problem of distinguishing small juntas from large DNF formulas, where ``small" and ``large" are functions of $k$ and $k'$ respectively: 

\begin{lemma}
\label{lem:Set-Cover-to-DTs} 
There is an algorithm that, given a size-$N$ instance $\mathcal{S}$ of $(k,k')$-{\sc Set-Cover} with $n$ sets and a parameter $\ell\le N$, runs in $\poly(N)$ time and outputs a circuit representation of a function $f : (\zo^\ell)^n \to\zo$ and a generator for a distribution $\mathcal{D}$ over $(\zo^\ell)^n$ satisfying:
\begin{itemize} 
\item[$\circ$] If $\opt(\mathcal{S}) \le k$, then $f$ is a  $k\ell$-junta under $\mathcal{D}$. 
\item[$\circ$] If $\opt(\mathcal{S}) > k'$, then any DNF of size $\le \exp(O(k'\ell))$ is $\Omega(\frac1{N})$-far from $f$ under $\mathcal{D}$.  
\end{itemize} 
\end{lemma}

We obtain~\Cref{thm:DT-DT,thm:junta-DNF,thm:proper} by combining~\Cref{lem:Set-Cover-to-DTs} with a recent result on the ETH-hardness of $(k,k')$-{\sc Set-Cover} for $k' = \frac1{2}\left(\frac{\log N}{\log\log N}\right)^{1/k}$, where $N$ is the size of the instance~\cite{Lin19}.  Similarly, we obtain~\Cref{thm:opt} by combining~\Cref{lem:Set-Cover-to-DTs} with~\Cref{conj:optimal-SetCover-hardness}.  For~\Cref{thm:junta-junta}, we only need a simpler special case of~\Cref{lem:Set-Cover-to-DTs}, which we combine with the ETH- and SETH-hardness of $(k,k+1)$-{\sc Set-Cover}   (i.e.~the hardness of solving parameterized {\sc Set-Cover} {\sl exactly})~\cite{CHKX06,PW10}.

\begin{figure}[ht]
    \centering
    \begin{tikzpicture}[]
        \def\x{3.5}
        \def\c{0}
        \def\l{-0.1}
        \def\r{0.1}
        \draw[black,xshift=-\x cm,|-|] (\c,-3) node[left] {$0$} -- (\c,0) node[left]{$n$};
        \draw[black, |-|] (\c,-3) node[right] {$0$} -- (\c,2) node[right]{$2^n$};
        \draw[black,xshift=\x cm, |-|] (\c,-3) node[right] {$0$} -- (\c,2) node[right]{$2^n$};
        \draw[color=black,xshift=-\x cm] (\l,-2) node[left,fill=white] {{\color{black}$\opt \le k$}};
        \node[draw,circle,fill=black,inner sep=1.5pt,xshift=-\x cm] (k) at (\c,-2) {};
        \draw[color=black,xshift=-\x cm] (\l,-1) node[left,fill=white] {\color{black} $\opt >k'$};
        \node[draw,circle,fill=black,inner sep=1.5pt,xshift=-\x cm] (k') at (\c,-1) {};
        \node[draw,circle,fill=black,inner sep=1.5pt,xshift=0cm] (k'M) at (\c,-1) {};
        \node[draw,circle,fill=black,inner sep=1.5pt,xshift=0cm] (kM) at (\c,-2) {};
        \draw[] ([yshift=.15cm,xshift=0cm]k'M) node[anchor=north west] {\footnotesize{size-$k'$ DT}};
        \draw[] ([yshift=.05cm,xshift=0cm]kM) node[anchor=north west] {\footnotesize{size-$k$ DT}};

        \draw[color=black,xshift=\x cm] (\r,1) node[right,fill=white] {size-$2^{\Omega({\color{black} k'}\ell)}$ DNF};
        \node[draw,circle,fill=black,inner sep=1.5pt,xshift=\x cm] (k'l) at (\c,0.92) {};
        \draw[color=black,xshift=\x cm] (\r,-1.25) node[right,fill=white] {$k\ell$-junta};
        \node[draw,circle,fill=black,inner sep=1.5pt,xshift=\x cm] (kl) at (\c,-1.25) {};
        \draw[xshift=1.85 cm] (\c,1.75) node[below,fill=white] {\footnotesize (Gap amplification)};
        \draw[xshift=-1.85 cm] (\c,1.75) node[below,fill=white] {\footnotesize (Easy reduction)};
        \draw[black,-stealth] (kM) to node[midway,below,sloped] {} (kl);
        \draw[black,-stealth] (k'M) to node[midway,above,sloped] {} (k'l);
        \draw[black,-stealth] (k') to node[midway,above,sloped] {} (k'M);
        \draw[black,-stealth] (k) to node[midway,below,sloped] {} (kM);
        \draw[xshift=-\x cm] (\c,-3.75) node [fill=white] {Set cover size};
        \draw[xshift=1.75cm] (\c,-3.75) node [fill=white] {Complexity of $f$ under $\mathcal{D}$};
        \draw[] (-0.75,3.5) node [left,text width=0.3\textwidth,rounded corners,fill=gray!10,inner sep=1ex, text centered] {
            $(k,k')$-{\sc Set-Cover} on $n$ vertices requires time {\color{black} $t(n,k)$}
        };
        \draw[] (0.75,3.5) node [right,text width=0.3\textwidth, rounded corners,fill=gray!10,inner sep=1ex, text centered] {
            Learning juntas with DNF hypotheses requires time $\min\{{\color{black} t(n,k)},2^{\Omega(k'\ell)}\}$
        };
        \draw[] (0,3.5) node[fill=white] {{\LARGE $\Rightarrow$}};
    \end{tikzpicture}
    \caption{An illustration of~\Cref{lem:Set-Cover-to-DTs} as a gap amplification technique. We take an instance of $(k,k')$-{\sc Set-Cover} where the gap between $k$ and $k'$ is small and first construct a distribution and a function whose decision tree complexity under the distribution exactly reflects the set cover gap. Then we amplify the distribution and the function to obtain an even more drastic gap in the complexity of the function under the distribution. \Cref{lem:Set-Cover-to-DTs} is quite versatile and underlies the proofs of \Cref{thm:DT-DT,thm:junta-DNF,thm:opt,thm:test,thm:proper}.}
    \label{fig:decision-tree-gap-amplification}
\end{figure}
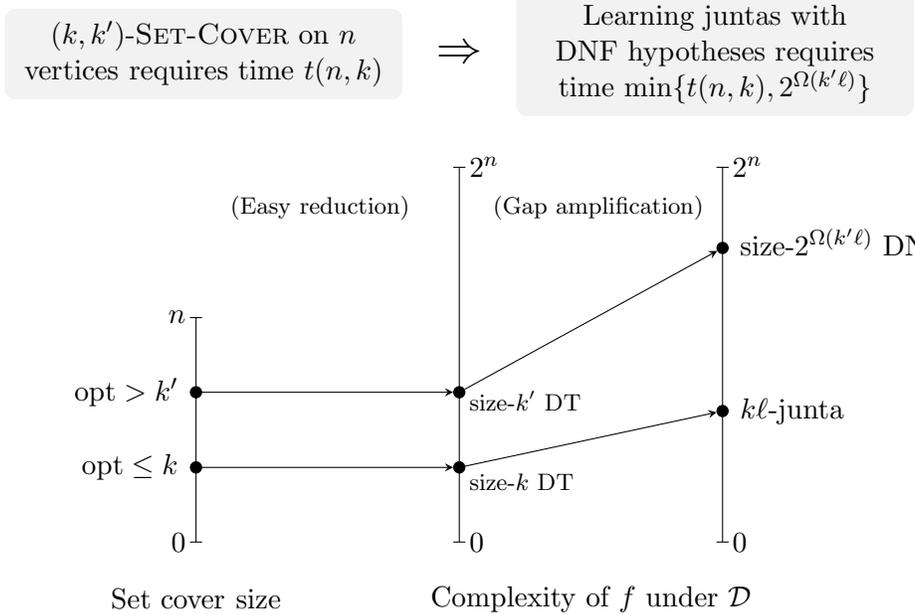

\paragraph{Gap amplification.}{We view \Cref{lem:Set-Cover-to-DTs} as a gap amplification procedure. Specifically, given a $(k,k')$-{\sc Set-Cover} instance it is straightforward to construct an instance of {\sc DT-Construction}, a target function $f$ and distribution $\mathcal{D}$, where the decision tree complexity of $f$ under $\mathcal{D}$ exactly reflects the gap $(k,k')$: if $\opt(\mathcal{S})\le k$ then $f$ is a size-$k$ decision tree under $\mathcal{D}$, and otherwise $f$ requires decision trees of size $\ge k'$. To obtain stronger lower bounds we amplify this gap into a much larger gap in the complexity of $f$ under $\mathcal{D}$: if $\opt(\mathcal{S})\le k$ then $f$ is a small junta under $\mathcal{D}$, and if $\opt(\mathcal{S})>k'$ then $f$ is a large-size DNF under $\mathcal{D}$. This reduction enables us to translate lower bounds for $(k,k')$-{\sc Set-Cover} into strong lower bounds for {\sc DT-Construction}. See \Cref{fig:decision-tree-gap-amplification} for an illustration of this gap amplification. 
}

\paragraph{Building hard instances of {\sc DT-Construction}.}{Our construction of $f$ and $\mathcal{D}$ in \Cref{lem:Set-Cover-to-DTs} is based on the one in~\cite{ABFKP09}, which in turn builds on~\cite{Hau88,HJLT96}. \cite{ABFKP09} also gave a gap amplifying reduction from $(k,k')$-{\sc Set-Cover} to the problem of distinguishing whether $f$ has small or large decision tree complexity under $\mathcal{D}$. \Cref{lem:Set-Cover-to-DTs} is a strengthening of their reduction where the same gap in set cover sizes leads to a more dramatic gap in $f$'s complexity under~$\mathcal{D}$.  While the construction of $f$ and $\mathcal{D}$ is similar to the one in~\cite{ABFKP09}, our analysis is entirely different and is, in our opinion, simpler. Notably, our analysis enables us to obtain lower bounds even against DNF hypotheses whereas previous works relied crucially on the hypothesis being a decision tree. In addition to yielding our stronger conclusion, our analysis overcomes technical challenges that arise when we have to modify~\Cref{lem:Set-Cover-to-DTs} in the context of {\sc DT-Estimation}, which we now discuss.}


\paragraph{Hardness amplification using XOR lemmas for query complexity.} For our lower bounds for {\sc DT-Estimation}, we begin by  observing that~\Cref{lem:Set-Cover-to-DTs}, when combined with existing results on the inapproximability of {\sl non}parameterized {\sc Set-Cover}~\cite{DS14,Mos15}, already implies a mild form of hardness of {\sc DT-Estimation}:  

\begin{corollary}[Mild hardness for {\sc DT-Estimation}] 
\label{cor:mild-hardness} 
Under ETH, there is no $\exp(d^{\,\Omega(1)})$ time algorithm that, given as input the circuit representation of a function $f: \zo^n \to \zo$, a generator for a distribution $\mathcal{D}$ over $\zo^n$, a parameter $d\in \N$, distinguishes between: 
\begin{itemize} 
\item[$\circ$] {\sc Yes}: $f$ is a depth-$d$ decision tree under $\mathcal{D}$. 
\item[$\circ$] {\sc No}: $f$ is $\Omega(\frac1{n})$-far from every decision tree of depth $\Omega(d\log d)$ under $\mathcal{D}$.
\end{itemize} 
\end{corollary} 

We amplify this mild hardness ($\eps = O(\frac1{n})$) to very strong hardness ($\eps = $ exponentially close to $\frac1{2}$) by considering $f^{\oplus m}: (\zo^{n})^m\to\zo$, the $m$-fold XOR composition of $f$: 
\[ f^{\oplus m}(x^{(1)},\ldots,x^{(m)}) \coloneqq f(x^{(1)})\oplus \cdots \oplus f(x^{(m)})\]
and the corresponding distribution $\mathcal{D}^{m}$ over $(\zo^n)^m$.  In the {\sc Yes} case of~\Cref{cor:mild-hardness}, it is easy to see that $f^{\oplus m}$ is a decision tree of depth $\le dm$ under $\mathcal{D}^m$.  To analyze the {\sc No} case, we prove the following lemma:

\begin{lemma}[Hardness amplification for {\sc DT Estimation}] 
\label{lem:hardness-amplification} 
Let $f : \zo^n \to \zo$ and $\mathcal{D}$ be such that $f$ is $\eps$-far from every depth-$d$  decision tree under $\mathcal{D}$. For any $\gamma > 0$, by taking $m = \Theta(\log(1/\gamma)/\eps)$, we get that $f^{\oplus m}$ is $(\frac1{2} - \gamma)$-far from every decision tree of depth $\Omega(dm)$ under $\mathcal{D}^m$.
\end{lemma}

Our proof of~\Cref{lem:hardness-amplification} combines existing XOR lemmas for distributional query complexity~\cite{Dru12,BB19,BKLS20}.  Specifically, we first use one due to Brody, Kim, Lerdputtipongporn, and Srinivasulu~\cite{BKLS20} to amplify from $\eps = O(\frac1{n})$ to $\Theta(1)$, and then one due to Drucker~\cite{Dru12} to amplify from $\Theta(1)$ to exponentially close to $\frac1{2}$.  The quantitative parameters of these lemmas are incomparable, and we show how they can be applied in tandem in our setting.

\paragraph{Handling aborts.} \Cref{lem:hardness-amplification} as stated is actually not quite what we prove; see~\Cref{lem:hardness-amplification-actual} for the actual version.  For technical reasons, the XOR lemma of~\cite{BKLS20} (and hence~\Cref{lem:hardness-amplification}) requires a stronger assumption, that $f$ is $\eps$-far from every depth-$d$ decision tree  that is allowed to {\sl abort} with probability $\delta$, and distance is measured with respect to {\sl non-aborts}.  \cite{BKLS20}'s lemma requires $\delta = \Theta(1)$ whereas $\eps = O(\frac1{n})$ in our setting, so this is a significantly stronger assumption.  To satisfy this stronger assumption, we have to prove a strengthening of~\Cref{cor:mild-hardness} where the {\sc No} case maps to an $f$ that is $\Omega(\frac1{n})$-far from decision trees that are allowed to abort with constant probability; this in turn necessitates a corresponding strengthening of~\Cref{lem:Set-Cover-to-DTs}.  With these in hand,~\Cref{thm:test} then follows fairly easily. 


\section{Discussion and future work} 

Our work makes new progress on the longstanding open problem of determining the complexity of properly PAC learning decision trees.  A natural avenue for future work is to close the remaining gap between our lower bound of $n^{\tilde{\Omega}(\log\log s)}$ and the $n^{O(\log s)}$ runtime of Ehrenfeucht and Haussler's  algorithm.  Our techniques point to an approach towards an $n^{\Omega(\log s)}$ lower bound via~\Cref{conj:optimal-SetCover-hardness}, which adds further motivation to the study of parameterized {\sc Set-Cover}.

As for our testing lower bounds, a notable feature is that they hold in the regime where $\eps = \frac1{2}-o(1)$, which we obtain from an initial hardness for $\eps = \Omega(\frac1{n})$ 
via XOR lemmas for query complexity.  It would be interesting to further develop such hardness amplification techniques in property testing.  For example, can the communication-complexity-based lower bound technique of Blais, Brody, and Matulef~\cite{BBM12} be fruitfully combined with the large body of work on XOR lemmas, and direct-product-type results more generally, for communication complexity?

More broadly, there is a growing and concerted effort within the machine learning community to design algorithms that produce {\sl simple} hypotheses, such as decision trees, especially in the context of high-stakes applications where interpretability is paramount; see e.g.~the position paper~\cite{Rud19}.  Our lower bounds show that interpretability can come at the price of computational intractability, even under strong assumptions on the target function.  There is  substantial practical motivation for the development of a theoretical understanding of such tradeoffs and how they can be mitigated.  For example, a concrete next step from our work is to identify reasonable assumptions under which our lower bounds can be circumvented; one could consider monotone target functions, a common assumption in both theory and practice.

\section{Preliminaries}
\paragraph{Set Cover.}{
Given a bipartite graph $\mathcal{S}=(S,U,E)$ on $N$-vertices, the \setcover problem is to find a minimum size subset $C\sse S$ such that every vertex in $U$ is adjacent to some vertex in $C$.\footnote{Typically, the set cover problem is cast as a combinatorial problem: given subsets $S_1,\ldots,S_m\sse [n]$ of some universe $[n]$, find the minimum size subcollection $S_{i_1},\ldots,S_{i_k}$ whose union is $[n]$. We consider the graph theoretic formulation because it makes the connection to the hitting set problem more transparent.}
We write $\opt(\mathcal{S})\in\N$ to denote the size of the smallest set cover for $\mathcal{S}$. We will often write $n$ to denote the size of $|S|\le N$. The set of neighbors of a vertex $u\in U$ is $\mathcal{N}_\mathcal{S}(u)=\{s\in S:(s,u)\in E\}$. We identify a vertex $u\in U$ with its neighborhood set $\mathcal{N}_\mathcal{S}(u)$. Each set $\mathcal{N}_\mathcal{S}(u)$ can be viewed as a string in $\zo^{|S|}$ where a $1$ in the string indicates an edge between $u$ and the corresponding vertex $s\in S$.  Hence, each vertex $u\in U$ can likewise be encoded as a string in $\zo^{|S|}$.\footnote{We assume without loss of generality that each $\mathcal{N}_\mathcal{S}(u)$ is unique so that a vertex $u$ can be identified by its neighborhood set $\mathcal{N}_\mathcal{S}(u)$ (if $\mathcal{N}_\mathcal{S}(u)=\mathcal{N}_\mathcal{S}(u')$ for $u\neq u'$ we can simply delete $u'$ without affecting the set cover complexity)}
}

\paragraph{Hitting Set.}{
Given a bipartite graph $\mathcal{H}=(S,U,E)$, the {\sc Hitting-Set} problem is to find a minimum size subset $I\sse U$ which ``hits'' every vertex $s\in S$: $\mathcal{N}_{\mathcal{S}}(s)\cap I\neq \varnothing$ for all $s\in S$. We write $\opt(\mathcal{H})$ for the size of the smallest hitting set.
}

An instance $\mathcal{H}=(S,U,E)$ of {\sc Hitting-Set} can equivalently be viewed as an instance $\mathcal{H}=(U,S,E)$ of {\sc Set-Cover}. 
\begin{fact}[{\sc Set-Cover} and {\sc Hitting-Set} are equivalent]
{\sc Set-Cover} and {\sc Hitting-Set} are equivalent to each other under approximation-preserving reductions. In particular, any instance $\mathcal{S}$ of {\sc Set-Cover} can be transformed in linear-time into an instance $\mathcal{H}$ of hitting set such that $\opt(\mathcal{S})=\opt(\mathcal{H})$ and vice versa.
\end{fact}

The results of \cite{ABFKP09} are formulated in terms of hitting set. Though for consistency, in this work we will only refer to {\sc Set-Cover}. See \Cref{fig:bipartite-graph-set-cover-hitting-set} for an illustration of a set cover instance and a hitting set instance on a single bipartite graph.

\begin{figure}[H]
\centering
\begin{subfigure}{.5\textwidth}
  \centering
  \resizebox{0.65\textwidth}{!}{%
      \begin{tikzpicture}[
      every node/.style={draw,circle,minimum size=0.25cm,inner sep=-0pt},
      usnode/.style={fill=black},
      ssnode/.style={fill=black},
      every fit/.style={ellipse,draw,inner sep=-2pt,text width=2cm},
      ->,shorten >= 2pt,shorten <= 2pt
    ]
    \begin{scope}[start chain=going below,node distance=7mm]
    \node[ssnode,on chain,teal] (s1) [label=left: {}] {};
    \node[ssnode,on chain,teal] (s2) [label=left: {}] {};
    \node[ssnode,on chain,teal] (s3) [label=left: {}] {};
    \node[ssnode,on chain] (s4) [label=left: {}] {};
    \node[ssnode,on chain] (s5) [label=left: {}] {};
    \end{scope}
    
    \begin{scope}[xshift=4cm,yshift=-0.5cm,start chain=going below,node distance=7mm]
    \node[usnode,on chain] (u1) [label=right: {}] {};
    \node[usnode,on chain] (u2) [label=right: {}] {};
    \node[usnode,on chain] (u3) [label=right: {}] {};
    \node[usnode,on chain] (u4) [label=right: {}] {};
    \end{scope}
    
    \node [black,fit=(s1) (s5),label={[yshift=0.75em]$S$}] {};
    \node [black,fit=(u1) (u4),label={[yshift=0.75em]$U$}] {};
    
    \node[coordinate] (S1) at (s1){};
    \node[coordinate] (S2) at (s2){};
    \node[coordinate] (S3) at (s3){};
    \node[coordinate] (S4) at (s4){};
    \node[coordinate] (S5) at (s5){};
    \node[coordinate] (U1) at (u1){};
    \node[coordinate] (U2) at (u2){};
    \node[coordinate] (U3) at (u3){};
    \node[coordinate] (U4) at (u4){};
    
    \draw[-,teal,very thick] (s1.center) -- (u1.center);
    \draw[-,teal,very thick] (S1) -- (U2);
    \draw[-,teal,very thick] (S1) -- (U3);
    \draw[-,teal,very thick] (S2) -- (U4);
    \draw[-,teal,very thick] (S3) -- (U3);
    \draw[-,teal,very thick] (S3) -- (U2);
    \draw[-] (S4) -- (U3);
    \draw[-] (S5) -- (U4);
        
    \node[draw,circle,fill=black,minimum size=0.25cm,inner sep=-0pt] at (U1) {};
    \node[draw,circle,fill=black,minimum size=0.25cm,inner sep=-0pt] at (U2) {};
    \node[draw,circle,fill=black,minimum size=0.25cm,inner sep=-0pt] at (U3) {};
    \node[draw,circle,fill=black,minimum size=0.25cm,inner sep=-0pt] at (U4) {};

    \end{tikzpicture}
    }
  \caption{A set cover of size 3 for $G$ highlighted in teal}
  \label{fig:set-cover}
\end{subfigure}%
\begin{subfigure}{.5\textwidth}
  \centering
  \resizebox{0.65\textwidth}{!}{
          \begin{tikzpicture}[
          every node/.style={draw,circle,minimum size=0.25cm,inner sep=0pt},
          usnode/.style={fill=black},
          ssnode/.style={fill=black},
          every fit/.style={ellipse,draw,inner sep=-2pt,text width=2cm},
          ->,shorten >= 2pt,shorten <= 2pt
        ]
        \begin{scope}[start chain=going below,node distance=7mm]
        \node[ssnode,on chain] (s1) [label=left: {}] {};
        \node[ssnode,on chain] (s2) [label=left: {}] {};
        \node[ssnode,on chain] (s3) [label=left: {}] {};
        \node[ssnode,on chain] (s4) [label=left: {}] {};
        \node[ssnode,on chain] (s5) [label=left: {}] {};
        \end{scope}
        
        \begin{scope}[xshift=4cm,yshift=-0.5cm,start chain=going below,node distance=7mm]
        \node[usnode,on chain] (u1) [label=right: {}] {};
        \node[usnode,on chain] (u2) [label=right: {}] {};
        \node[usnode,on chain,violet] (u3) [label=right: {}] {};
        \node[usnode,on chain,violet] (u4) [label=right: {}] {};
        \end{scope}
        
        \node [black,fit=(s1) (s5),label={[yshift=0.75em]$S$}] {};
        \node [black,fit=(u1) (u4),label={[yshift=0.75em]$U$}] {};
        \node[coordinate] (S1) at (s1){};
        \node[coordinate] (S2) at (s2){};
        \node[coordinate] (S3) at (s3){};
        \node[coordinate] (S4) at (s4){};
        \node[coordinate] (S5) at (s5){};
        \node[coordinate] (U1) at (u1){};
        \node[coordinate] (U2) at (u2){};
        \node[coordinate] (U3) at (u3){};
        \node[coordinate] (U4) at (u4){};
        \draw[-] (S1) -- (U1);
        \draw[-] (S1) -- (U2);
        \draw[-,violet,very thick] (S1) -- (U3);
        \draw[-,violet,very thick] (S2) -- (U4);
        \draw[-,violet,very thick] (S3) -- (U3);
        \draw[-] (S3) -- (U2);
        \draw[-,violet,very thick] (S4) -- (U3);
        \draw[-,violet,very thick] (S5) -- (U4);
        
        \node[draw,circle,fill=black,minimum size=0.25cm,inner sep=-0pt] at (S1) {};
        \node[draw,circle,fill=black,minimum size=0.25cm,inner sep=-0pt] at (S2) {};
        \node[draw,circle,fill=black,minimum size=0.25cm,inner sep=-0pt] at (S3) {};
        \node[draw,circle,fill=black,minimum size=0.25cm,inner sep=-0pt] at (S4) {};
        \node[draw,circle,fill=black,minimum size=0.25cm,inner sep=-0pt] at (S5) {};
        
        \end{tikzpicture}
    }
  \caption{A hitting set of size $2$ for $G$ highlighted in purple}
  \label{fig:hitting-set}
\end{subfigure}
\caption{A bipartite graph $G=(S,U,E)$ viewed on the left as a set cover instance and on the right as a hitting set instance.}
\label{fig:bipartite-graph-set-cover-hitting-set}
\end{figure}
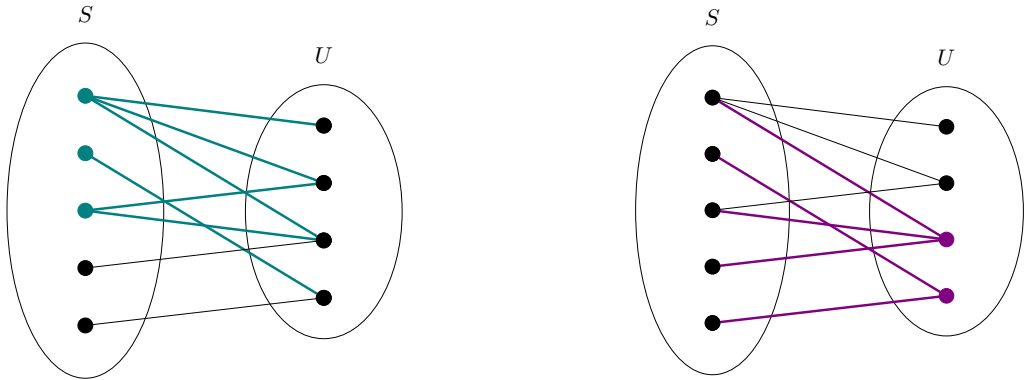

\paragraph{Decision trees.}{For a decision tree $T:\zo^n\to\zo$, we write $L\in T$ to denote that $L$ is a leaf of $T$. The size of $T$ is its number of leaves and is denoted $|T|$. For an input $x\in\zo^n$, we write $\depth_T(x)\in \N$ to denote the \textit{depth} of $x$ in $T$, the number of variables queried on the root-to-leaf path consistent with $x$. }

\paragraph{DNF formulas.}{A literal is a variable or its negation. A term is a conjunction ($\land$) of literals. A DNF formula $F:\zo^n\to\zo$ is a disjunction ($\lor$) of terms, denoted $F=t_1\lor\cdots\lor t_s$. The size of the DNF formula is $|F|=s$, the number of terms. The \textit{width} of a term $|t_i|$ is the number of literals in it. The width of an input $x\in\zo^n$ is defined as the width of the smallest width term accepting $x$ and $0$ if no term accepts $x$:
$$
\width_F(x)\coloneqq
\begin{cases}
\underset{t_i(x)=1}{\min} |t_i| & F(x)=1\\
0 & F(x)=0.
\end{cases}
$$
}

\paragraph{Circuits.}{We consider Boolean circuits $\mathcal{C}:\zo^n\to\zo$ with AND, OR, NOT, and PARITY gates: $\{\land,\lor,\lnot,\oplus\}$. The size of a circuit $|\mathcal{C}|$ is the number of gates in it. The depth of a circuit is the longest directed path from an input node to an output node.
}

\paragraph{$k$-juntas.}{A function $f:\zo^n\to\zo$ is a $k$-junta if its output depends on $\le k$ bits. Hence, if $f$ is a $k$-junta it can be completely specified by a table of size $2^k$ corresponding to all possible assignments to the $k$ relevant variables. In particular, every $k$-junta is a size-$2^k$ decision tree and every size-$s$ decision tree is an $s$-junta.}

\paragraph{Distributions.}{We use boldface letters e.g. $\bx,\by$ to denote random variables. For a distribution $\mathcal{D}$, we write $\dist_{\mathcal{D}}(f,g)=\Pr_{x\sim \mathcal{D}}[f(\bx)\neq g(\bx)]$. A function $f$ is $\eps$-close to $g$ if $\dist_{\mathcal{D}}(f,g)\le \eps$. When $\eps=0$, we drop the $\eps$ and simply say $f$ computes $g$ over $\mathcal{D}$. Often $f$ is viewed as one of the combinatorial objects above and $g$ is a generic function, e.g. a decision tree $T:\zo^n\to\zo$ computes $g$ over $\mathcal{D}$ if $\dist_{\mathcal{D}}(T,g)=0$. Similarly, $f$ is $\eps$-far from $g$ if $\dist_{\mathcal{D}}(f,g)>\eps$. We write $\mathcal{U}_b$ for the uniform distribution on $b$ bits. A generator for a distribution $\mathcal{D}$ over $\zo^n$ is an algorithm $G:\zo^n\to\zo^n$ which takes $n$ uniform random bits as input and outputs $n$ bits distributed according to $\mathcal{D}$: $\Pr_{\bx\sim\mathcal{U}_n}[G(\bx)=x]=\Pr_{\bx\sim\mathcal{D}}[\bx=x]$ for all $x \in \zo^n$.
}

\paragraph{Learning.}{See \Cref{appendix:pac-learning} for the definitions of learning that we use. All learning algorithms we consider are \textit{proper} learning algorithms. When referring to ``learning decision trees'' we mean properly learning the concept class $\mathcal{T}=\{T:\zo^n\to\zo\mid T\text{ is a decision tree}\}$. Likewise, when referring to ``learning size-$s$ decision trees'', we mean properly learning the concept class $\mathcal{T}_s=\{T:\zo^n\to\zo\mid T\text{ is a size-}s\text{ decision tree}\}$. When discussing algorithms for learning $k$-juntas, we assume the output of the learning algorithm is a table of size $2^k$ (as in e.g. \cite{MOS04}). 
}

\paragraph{Complexity-theoretic assumptions.}{
Many results on the hardness of {\sc Set-Cover} are conditioned on the exponential time hypothesis.

\begin{hypothesis}[Exponential time hypothesis (ETH) \cite{Tovey84, IP01, IPZ01}]
There exists a constant $\delta>0$ such that $3$-SAT on $n$ variables cannot be solved in $O(2^{\delta n})$ time.
\end{hypothesis}

Since we are proving hardness against randomized algorithms, we will use a randomized variant of ETH.

\begin{hypothesis}[Randomized ETH, see \cite{CIKP03,DHHTMTW14}]
There exists a constant $\delta>0$ such that $3$-SAT on $n$ variables cannot be solved by a randomized algorithm in $O(2^{\delta n})$ time with error probability at most $1/3$.
\end{hypothesis}

We will also use two additional hypotheses.

\begin{hypothesis}[Strong exponential time hypothesis (SETH) \cite{IP01, IPZ01}]
For every $\delta>0$, there exists a $k\in\N$ such that $k$-CNF-SAT on $n$ variables cannot be solved in time $O(2^{n(1-\delta)})$.
\end{hypothesis}

\begin{hypothesis}[{$W[1]\neq \textnormal{FPT}$}, see \cite{DF13,CyganFKLMPPS15}]
For any computable function $f:\N\to\N$, no algorithm can decide if a graph $G=(V,E)$ contains a $k$-clique in $f(k)\cdot \poly(|V|)$ time. 
\end{hypothesis}

As with randomized ETH, randomized SETH and randomized $W[1]\neq \textnormal{FPT}$ are the respective versions of these hypotheses against randomized algorithms. Also, we remark that $W[1]\neq \textnormal{FPT}$ is a weaker assumption than ETH which itself is weaker than SETH. If $W[1]=\textnormal{FPT}$, then SAT is solvable in subexponential time.
}
\subsection{Existing results on the hardness of {\sc Set-Cover}}
\label{subsection:results-on-set-cover}
Throughout, we use several different hardness results for {\sc Set-Cover} and approximating {\sc Set-Cover}. We start with the following theorem due to \cite{Lin19} about the hardness of approximating set cover. We have slightly modified the theorem from its original form to fit our setting. We discuss Lin's original theorem and our modifications in \Cref{appendix:set_cover_hardness}.
\begin{theorem}[\cite{Lin19}]
\label{thm:SC_hardness}
Assuming randomized ETH, there is a constant $c\in(0,1)$ such that for any $k\in\N$ with $k\le \frac{1}{2}\cdot\frac{\log\log N}{\log\log\log N}$, there is no randomized $N^{c k}$ time algorithm that can solve $\left(k,\frac{1}{2}\left(\frac{\log N}{\log\log N}\right)^{1/k}\right)$-\setcover on $N$ vertices with high probability.
\end{theorem}

We will also use results on the inapproximability of {\sl unparameterized} {\sc Set-Cover}:   

\begin{theorem}[\cite{DS14,Mos15}] 
\label{thm:nonparameterized-SC-1} 
Under randomized ETH, for every $0 < \beta < 1$, any  algorithm that approximates size-$N$ instances of {\sc Set-Cover} to within $(1-\beta)\ln N$ w.h.p.~requires $2^{N^{\Omega(\beta)}}$ time.
\end{theorem} 

By a standard search-to-decision reduction,~\Cref{thm:nonparameterized-SC-1} implies the following lower bound for $(k,k')$-\setcover where, unlike in the parameterized setting, $k$ is no longer guaranteed to be ``small": 

\begin{theorem} 
\label{thm:nonparameterized-SC} 
Under randomized ETH, for every $0 < \beta < 1$, there exists $k \le N$ such that any algorithm that solves  size-$N$ instances of $(k,k')$-\setcover where $k' = k(1-\beta)\ln N $ w.h.p.~requires $2^{N^{\Omega(\beta)}}$ time.
\end{theorem}

Finally, we will also use existing lower bounds in the {\sl ungapped} setting:

\begin{theorem}[Ungapped hardness of {\sc Set-Cover} from {$W[1]\neq \textnormal{FPT}$} {\cite[Theorem 5.6]{CHKX06}}]
\label{thm:FPT-hardness-of-ungapped-set-cover}
Assuming $W[1]\neq \textnormal{FPT}$, for all constants $c\in (0,1)$ and for all $k\le n^c$, any $(k,k+1)$-{\sc Set-Cover} instance $\mathcal{S}=(S,U,E)$ cannot be solved in time $|S|^{o(k)}$.
\end{theorem}

Furthermore, there are even stronger set cover lower bounds assuming SETH.

\begin{theorem}[Ungapped hardness of {\sc Set-Cover} from SETH {\cite[Theorem 2.3]{PW10}}]
\label{thm:SETH-hardness-of-ungapped-set-cover}
Assuming SETH, for all constants $c,\delta\in (0,1)$ and for all $k\le n^c$, any $(k,k+1)$-{\sc Set-Cover} instance $\mathcal{S}=(S,U,E)$ cannot be solved in time $O(|S|^{k-\delta})$. 
\end{theorem}

\section{Lower bounds for {\sc DT-Construction}}
\label{sec:DT-construction} 

In this section we prove~\Cref{lem:Set-Cover-to-DTs} and use it to derive~\Cref{thm:junta-DNF,thm:junta-junta}.  The high-level idea behind~\Cref{lem:Set-Cover-to-DTs} is to show how, given a set cover instance $\mathcal{S}$, we can construct a function $f$ and a distribution $\mathcal{D}$ such that the optimal set cover size for $\mathcal{S}$ is reflected in the the complexity of $f$ under $\mathcal{D}$.


\begin{definition}[$\Gamma_{\mathcal{S}}$ and $\mathcal{D}_{\mathcal{S}}$]
\label{def:D_S_and_Gamma_S}
Let $\mathcal{S}=(S,U,E)$ be a set cover instance with $|S|=n$.  We identify each universe element $u \in U$ with a vector $\zo^n$, the indicator vector of its neighborhood set $\mathcal{N}_{\mathcal{S}}(u)$ (i.e.~the indicator vector of the sets that contain $u$).  We define the partial function $\Gamma_{\mathcal{S}} : \zo^n \to \zo$ as follows: 
\[ \Gamma_{\mathcal{S}}(x) = 
\begin{cases} 
0 & \text{$x = 0^n$} \\
1 & \text{$x = u$, $u\in U$}. 
\end{cases} 
\] 
The distribution $\mathcal{D}_{\mathcal{S}}$ over the support of $\Gamma_{\mathcal{S}}$ is given by the pmf 
\[ \mathcal{D}_{\mathcal{S}}(x) = 
\begin{cases} 
\frac1{2} & \text{$x = 0^n$}  \\
\frac1{2|U|} & \text{$x= u, u \in U$}.
\end{cases} 
\] 
\end{definition} 
When $\mathcal{S}$ is clear from context we will drop the subscript and simply write $\Gamma$ and $\mathcal{D}$.  We observe that given any set cover $C \subseteq S$, the monotone disjunction of the variables in $C$ computes $\Gamma$ over~$\mathcal{D}$.  In particular, we have: 
\begin{fact} 
\label{fact:Gamma-is-small-disjunction} 
If $\opt(\mathcal{S}) \le k$ then $\Gamma$ is a monotone disjunction of $k$ variables under~$\mathcal{D}$. 
\end{fact} 

We now define a ``parity-amplified" version of $\Gamma$.  While $\Gamma$ is a function over the domain $\zo^n$, this new function will be over the domain $(\zo^{\ell})^{n}$ for some parameter $\ell\in \N$. 

\newcommand{\BlockwisePar}{\textnormal{BlockwisePar}}

\newcommand{\ParComplete}{\textnormal{ParComplete}}

\paragraph{Notation.}  For a string $y \in (\zo^{\ell})^n$, we write $y_i\in \zo^\ell$ to denote the $i$th block of $y$, and $(y_i)_j$ to denote the $j$th entry of the $i$th block. We define the function $\BlockwisePar: (\zo^\ell)^n \to \zo^n$: 
\[ 
\BlockwisePar(y) \coloneqq (\oplus y_1, \ldots , \oplus y_n), 
\] 
where $\oplus y_i$ denotes the parity of the bits in $y_i$. 
\smallskip

\begin{definition}[$\Gamma_{\oplus \ell}$ and $\mathcal{D}_{\oplus \ell}$]
For $\Gamma$ and $\mathcal{D}$ as defined in~\Cref{def:D_S_and_Gamma_S} and an integer $\ell \in \N$, we define the partial function $\Gamma_{\oplus \ell} : (\zo^\ell)^{n}\to \zo$, 
\[ \Gamma_{\oplus \ell}(y) = \Gamma(\BlockwisePar(y)). \] 
The distribution $\mathcal{D}_{\oplus \ell}$ over the support of $\Gamma_{\oplus \ell}$ is defined as follows: to sample from $\mathcal{D}_{\oplus \ell}$, 
\begin{enumerate} 
\item First sample $\bx \sim \mathcal{D}$. 
\item For each $i\in [n]$, sample $\by_i\sim \zo^\ell$ u.a.r.~among all strings satisfying $\oplus \by_i = \bx_i$. Equivalently, sample $\by \sim (\zo^\ell)^n$ u.a.r.~among all strings satisfying $\BlockwisePar(\by) = \bx$. 
\end{enumerate} 
\end{definition} 

\begin{fact}[Blockwise parity of $\mathcal{D}_{\oplus \ell}$ induces $\mathcal{D}$]
\label{fact:D_ell-induces-D} 
For $\by\sim \mathcal{D}_{\oplus \ell}$, we have that $\BlockwisePar(\by)$ is distributed according to $\mathcal{D}$. 
\end{fact}
\noindent
We have the following analogue of~\Cref{fact:Gamma-is-small-disjunction}: 
\begin{fact} 
\label{fact:Gamma_ell-is-small-junta}
If $\opt(\mathcal{S})\le k$ then $\Gamma_{\oplus \ell}$ is a $k\ell$-junta (a disjunction of $k$ many parities, each over $\ell$ variables) under $\mathcal{D}_{\oplus \ell}$. 
\end{fact} 

\paragraph{An equivalent way of sampling from $\mathcal{D}_{\oplus \ell}$.}  For our proof of~\Cref{lem:Set-Cover-to-DTs}, it will be useful for us consider a different, but equivalent, way of sampling from $\mathcal{D}_{\oplus \ell}$.   For $z\in (\zo^{\ell-1})^n$, $x\in \zo^n$, and $j\in [\ell]$, we write $\ParComplete_j(z,x)$ to denote the string $y \in (\zo^{\ell})^n$ where for each block $i\in [n]$, 
\begin{itemize} 
\item[$\circ$] All except the $j$th coordinate of $y_i \in \zo^\ell$ are filled in according to $z_i \in \zo^{\ell-1}$. 
\[ ((y_i)_1,\ldots,(y_i)_{j-1},(y_i)_{j+1},\ldots,(y_i)_{\ell}) = ((z_i)_1,\ldots,(z_i)_{\ell-1}).\] 
\item[$\circ$] The $j$th coordinate of $y_i$ is filled in with the unique bit so that $\oplus y_i = x_i$.  
\end{itemize} 

\paragraph{Example.}Consider $n=4$ and $\ell=3$ and $j=2$. Then, we can view $z=(z_1,\ldots,z_4)\in (\zo^2)^4$ as a $4\times 2$ matrix where the $i$th row is $z_i$. In this case, we may have for example: 
$$
z=
\begin{bmatrix}
\teal{1} & \teal{0}\\
\teal{0} & \teal{0} \\
\teal{1} & \teal{1} \\
\teal{1}  & \teal{0} 
\end{bmatrix}\quad
x=
\begin{bmatrix}
1\\
1\\
0\\
1
\end{bmatrix}
\quad
\longrightarrow
\quad
\ParComplete_j(z,x)=
\begin{bmatrix}
\teal{1} & \violet{\bf 0} & \teal{0}\\
\teal{0} & \violet{\bf 1} & \teal{0}\\
\teal{1} & \violet{\bf 0} & \teal{1}\\
\teal{1} & \violet{\bf 0} & \teal{0}
\end{bmatrix}.
$$
Note that the first and third columns of $\ParComplete_j(z,x)$, colored teal, are exactly the first and second columns of $z$ respectively, and that the second column of $\ParComplete_j(z,x)$, colored purple, is filled in so that parity of each row of matches the corresponding row of $x$. 
\smallskip 

\begin{definition}[The distribution $\mathcal{D}_{\oplus\ell}^j$]
\label{def:alternative_Dlj}
For $j\in[\ell]$, the distribution $\mathcal{D}_{\oplus\ell}^j$ is obtained via the following sampling procedure:  sample $\bx\sim \mathcal{D}$, $\bz \sim (\zo^{\ell-1})^n$ u.a.r., and output $\ParComplete_j(\bz,\bx)$. 
\end{definition}

The following proposition on the equivalence between $\mathcal{D}_{\oplus \ell}$ and $\mathcal{D}_{\oplus \ell}^j$ can be easily verified. We defer the calculation to \Cref{appendix:Dj-equals-D}.

\begin{proposition}[$\mathcal{D}_{\oplus\ell}^j$ is equivalent to $\mathcal{D}_{\oplus\ell}$]
\label{prop:Dj_equals_D}
For all $j\in[\ell]$ and $y\in (\zo^\ell)^n$, 
$$
\pryd{\by=y}=\prydj{\by=y}.
$$
\end{proposition}

\subsection*{Constructiveness of $\Gamma_{\oplus\ell}$ and $\mathcal{D}_{\oplus\ell}$}
We can efficiently compute both a circuit representation of $\Gamma_{\oplus\ell}$ and a generator for the distribution $\mathcal{D}_{\oplus\ell}$ from a given set cover instance.

\begin{lemma}[Constructiveness of $\Gamma_{\oplus\ell}$ and $\mathcal{D}_{\oplus\ell}$]
\label{lem:circuit-for-gamma-generator-for-D}
Let $\mathcal{S}=(S,U,E)$ be an $N$-vertex set cover instance with $|S|=n$ and let $\ell\le N$ be a parameter. Then there is an algorithm that runs in $\poly(N)$ time and outputs a circuit representation of $\Gamma_{\oplus\ell}$ over $\mathcal{D}_{\oplus\ell}$ and a generator for the distribution $\mathcal{D}_{\oplus\ell}$. 
\end{lemma}

\begin{figure}[ht]
\begin{center}
\forestset{
 default preamble={
 for tree={
  circle,
  minimum size=0.7cm,
  inner sep=0.1pt,
  draw,
  align=center,
  anchor=north,
  fill=white,
  l sep=1cm,
  s sep=0.5cm,
 }
 }
}
\begin{forest}
 [$\lor$,name=lormain, s sep=0.5cm, calign=center,
    [{},name=oplus1, s sep=0.1cm, inner sep=-0.0cm,circle split,
        [$(y_1)_1$,name=y11,tier=inputl,draw=white]
        [${}$,name=ldots21,draw=white,tier=inputl]
        [$\ldots$,name=ldots22,draw=white,tier=inputl]
        [${}$,name=ldots23,draw=white,tier=inputl]
        [$(y_1)_\ell$,name=y1l,tier=inputl,draw=white]
    ]
    [${}$,name=ldots1,draw=white,before computing xy={s=-35}]
    [$\ldots$,name=ldots2,draw=white,before computing xy={s=0}]
    [${}$,name=ldots3,draw=white,before computing xy={s=35}]
    [{},name=oplus2,s sep=0.1cm,inner sep=-0.0cm,circle split,
        [$(y_n)_1$,name=yn1,tier=inputl,draw=white]
        [${}$,name=ldots31,draw=white,tier=inputl]
        [$\ldots$,name=ldots32,draw=white,tier=inputl]
        [${}$,name=ldots33,draw=white,tier=inputl]
        [$(y_n)_\ell$,name=ynl,tier=inputl,draw=white]
    ]
 ]
 \draw[black] (-1.75,-0.75) .. controls +(south east:1cm) and +(south west:1cm) .. (1.75,-0.75);
 \draw[black] (-3.5,-2.75) .. controls +(south east:0.5cm) and +(south west:0.5cm) .. (-1.5,-2.75);
 \node[xshift=-0.65cm,yshift=0.25cm] at (-2,-0.75) {fan-in $n$};
 \node[xshift=-0.75cm,yshift=0.25cm] at (-3.6,-2.75) {fan-in $\ell$};
 \node[draw,black,circle split,rotate=90,inner sep=-0.0cm,minimum size=0.7cm] at (oplus2) {};
 \node[draw,black,circle split,rotate=90,inner sep=-0.0cm,minimum size=0.7cm] at (oplus1) {};
 \node[draw,black,circle,inner sep=-0.0cm,minimum size=0.7cm] at (lormain) {};
\end{forest}
\caption{A depth-$2$ circuit for $\Gamma_{\oplus\ell}$ consisting of one top gate that is an OR connected to $n$ PARITY gates, each of which is connected to a disjoint block of $\ell$ input variables.}
\label{fig:circuit-for-GammaL}
\end{center}
\end{figure}
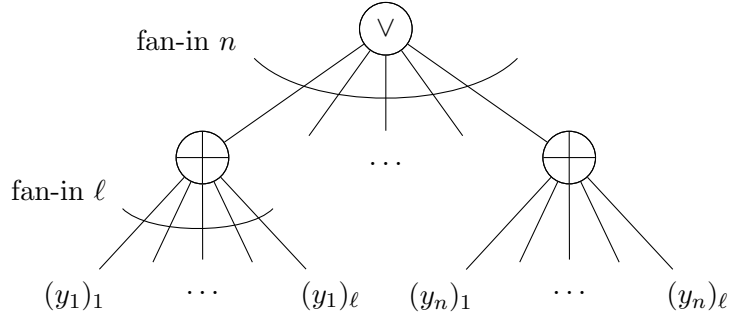

\begin{proof}
We separate the proof into two parts. First, we give a circuit representation of $\Gamma_{\oplus\ell}$, then we give a generator for $\mathcal{D}_{\oplus\ell}$.

\paragraph{A circuit for $\Gamma_{\oplus\ell}$.}{
Recall that a circuit $\mathcal{C}:(\zo^\ell)^{n}\to\zo$ represents $\Gamma_{\oplus\ell}:(\zo^\ell)^{n}\to\zo$ over $\mathcal{D}_{\oplus\ell}$ if $\dist_{\mathcal{D}_{\oplus\ell}}(\mathcal{C},\Gamma_{\oplus\ell})=0$. The function $\Gamma:\zo^n\to\zo$ is computed over $\mathcal{D}$ by the disjunction of all $n$ variables. That is, $\dist_{\mathcal{D}}(\Gamma,x_1\lor\cdots\lor x_n)=0$.\footnote{This observation can equivalently be viewed as an application of \Cref{fact:Gamma-is-small-disjunction} plus the fact that $\opt(\mathcal{S})\le |S|=n$ holds for all $\mathcal{S}$.} Therefore, for $y=(y_1,\ldots,y_n)\in\supp(\mathcal{D}_{\oplus\ell})$,
\begin{align*}
    \Gamma_{\oplus\ell}(y)&=\Gamma(\oplus y_1,\ldots,\oplus y_n)\tag{Definition of $\Gamma_{\oplus\ell}$}\\
    &=(\oplus y_1)\lor \ldots\lor (\oplus y_n)\tag{$\BlockwisePar(y)\in\supp(\mathcal{D})$}
\end{align*}
It follows that the circuit given by
$$
\mathcal{C}(y)\coloneqq \bigvee_{i\in[n]}\bigoplus_{j\in[\ell]} (y_i)_j
$$
computes $\Gamma_{\oplus\ell}$ over $\mathcal{D}_{\oplus\ell}$. See \Cref{fig:circuit-for-GammaL} for an illustration of $\mathcal{C}$. Since this circuit has size $n\cdot\ell$ and depth $3$, the first part of the lemma statement follows.
}

\paragraph{A generator for $\mathcal{D}_{\oplus\ell}$.}{
Recall that a generator for a distribution takes uniform random bits as input and outputs bits distributed according to the desired distribution. First, we observe that there is an efficient generator for $\mathcal{D}$ using $1+\log |U|$ uniform random bits. Specifically, use $1$ uniform random bit to decide between the two cases:
\begin{enumerate}[label=(\arabic*)]
    \item output $0^n$
    \item output $u\in U$ uniformly at random.
\end{enumerate}
The second case can be accomplished with $\log |U|$ uniform random bits. Then the following procedure generates the distribution $\mathcal{D}_{\oplus\ell}^1$: 
\begin{enumerate}[label=(\arabic*)]
    \item use $n(\ell-1)$ uniform random bits to select $z\in (\zo^{\ell-1})^n$
    \item use $1+\log |U|$ bits to sample $\bx\sim\mathcal{D}$
    \item output $\ParComplete_1(\bz,\bx)$.
\end{enumerate}
By \Cref{prop:Dj_equals_D}, this procedure equivalently generates the distribution $\mathcal{D}_{\oplus\ell}$. The procedure uses $n(\ell-1)+1+\log|U|$ bits. We can assume without loss of generality that $1+\log|U|\le |S|=n$\footnote{If $|S|<1+\log|U|$, we just replicate sets until $|U|\le |S|$. This change at most doubles $N$ and does not affect $\opt(\mathcal{S})$.} so that $n(\ell-1)+1+\log|U|\le n\ell$. It follows that this procedure efficiently generates $\mathcal{D}_{\oplus\ell}$ from $n\ell$ uniform random bits.  
}
\end{proof}

\subsection{Warmup for~\Cref{lem:Set-Cover-to-DTs}: Lower bounds against decision tree hypotheses}

We will prove~\Cref{lem:Set-Cover-to-DTs} with the function being $\Gamma_{\oplus \ell}$ and the distribution being $\mathcal{D}_{\oplus \ell}$.  The first bullet of the lemma statement is given by~\Cref{fact:Gamma_ell-is-small-junta}, and so the bulk of the remaining work goes into establishing the second bullet of the lemma statement.  

 We begin with a warmup, showing the weaker statement that $\Gamma_{\oplus \ell}$ is far from any small {\sl decision tree} under $\mathcal{D}_{\oplus \ell}$.  This proof will illustrate many of the key ideas in the actual proof for DNFs, which we give in the next subsection.  Furthermore, this lower bound is already sufficient to establish~\Cref{thm:DT-DT}, and will be the starting point of our lower bounds for {\sc DT-Estimation} that we prove in the next section.

\begin{lemma}
\label{lem:set_cover-reduction} 
Let $\mathcal{S}=(S,U,E)$ be an $N$-vertex set cover instance and let $\ell\ge 2$.  If $T : (\zo^\ell)^n \to \zo$ is a decision tree of size $|T|< 2^{\opt(\mathcal{S}) {\ell}/8}$, then  $\dist_{\mathcal{D}_{\oplus\ell}}(T,\Gamma_{\oplus\ell})\ge 1/(4N)$.
\end{lemma}

\paragraph{High level idea.}{
\label{para:high-level-idea-of-DT-gamma-lb}
There are three main steps:  
\begin{enumerate}[itemsep=0.1em]
    \item No decision tree with small average depth can approximate $\Gamma$ under $\mathcal{D}$ (\Cref{claim:depth_error-new}).
    \item Any decision tree with small average depth that approximates $\Gamma_{\oplus \ell}$ under $\mathcal{D}_{\oplus \ell}$ can be used to construct decision tree of much smaller average depth that approximates $\Gamma$ under $\mathcal{D}$ (\Cref{claim:extracting_small_dt-new}). This is the key claim. 
    \item Any small size decision tree must have small average depth with respect to $\mathcal{D}_{\oplus \ell}$ (\Cref{claim:avg_depth-new}). 
\end{enumerate}
Together, these three claims imply that no small size decision tree can approximate $\Gamma_{\oplus\ell}$ under $\mathcal{D}_{\oplus \ell}$, thereby yielding~\Cref{lem:set_cover-reduction}. 
}

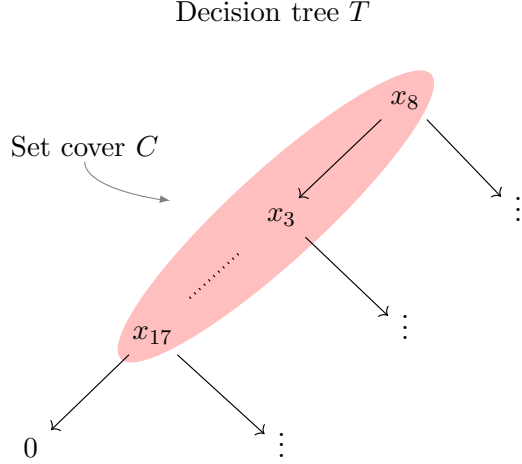
\begin{figure}[h]
\centering
\[\begin{tikzcd}[execute at end picture={
    \begin{scope}[on background layer]
        \fill[red!50,opacity=.5,rotate=42.5] (0.55,0.54) ellipse (2.8 and 0.6);
        \filldraw[black] (-2.5,1.7) node[] {Set cover $C$};
        \draw[-{latex reversed[length=3mm, width=2mm]},gray,tips=proper] (-2.5,1.5) edge [out=270, in=-8] (-1.5,1);
        \node[yshift=3.5cm] at (0,0) {Decision tree $T$};
    \end{scope}}]
	&&& {x_{8}} \\
	&& {x_{3}} && \vdots \\
	& {x_{17}} && \vdots \\
	0 && \vdots
	\arrow[from=1-4, to=2-3]
	\arrow[dotted, line width=1pt, no head, from=2-3, to=3-2, shorten=2ex]
	\arrow[from=3-2, to=4-1]
	\arrow[from=1-4, to=2-5]
	\arrow[from=2-3, to=3-4]
	\arrow[from=3-2, to=4-3]
\end{tikzcd}\]
\caption{Any decision tree for $\Gamma$ implicitly defines a set cover of $\mathcal{S}$ consisting of the variables highlighted in red.}
\label{fig:set-cover-from-DT-for-Gamma}
\end{figure}

\begin{claim}[Good approximators for $\Gamma$ require large depth]
\label{claim:depth_error-new}
Let $T : \zo^n \to \zo$ be a decision tree and $\mathcal{S}=(S,U,E)$ be an $N$-vertex set cover instance with $|S|=n$. If $\underset{\bx\sim \mathcal{D}}{\E}[\depth_T(\bx)]<\opt(\mathcal{S})/2$ then $\dist_{\mathcal{D}}(T,\Gamma)\ge 1/(2N)$.
\end{claim}

\begin{proof}
Let $T$ be a decision tree satisfying $\underset{\bx\sim D}{\E}[\depth_T(\bx)]<\opt(\mathcal{S})/2$. We actually prove the stronger claim that $\dist_{\mathcal{D}}(T,\Gamma)\ge 1/(2|U|)$. Suppose for contradiction that $\dist_{\mathcal{D}}(T,\Gamma)< 1/(2|U|)$. Each $x\in \supp(\mathcal{D})$ has mass $\ge 1/(2|U|)$ under $\mathcal{D}$ and so we must have $\dist_{\mathcal{D}}(T,\Gamma)=0$. Let $C\sse [n]=S$ be the set of vertices that $T$ queries in the computation of {$0^n$} (equivalently, $C$ is the {leftmost}  root-to-leaf path in $T$). See \Cref{fig:set-cover-from-DT-for-Gamma} for an illustration of $C$.  Since $\dist_{\mathcal{D}}(T,\Gamma) = 0$, we have that $T(0^n) = \Gamma(0^n) = 0$.

We claim $C$ is a valid set cover for $\mathcal{S}$. 
Indeed, if some $u \in U$ is not covered by $C$, then $\mathcal{N}_\mathcal{S}(u) \cap C = \varnothing$, and  $u$ would follow this same path $C$ as {$0^n$} in $T$.  This would imply that {$0 = T(u) \ne \Gamma(u) = 1$}, contradicting the fact that $\dist_{\mathcal{D}}(T,\Gamma) = 0$.  

Since $C$ is a valid set cover, it follows that $|C| \ge \opt(\mathcal{S})$ and so:   
\begin{align*}
    \underset{\bx\sim \mathcal{D}}{\E}[\depth_T(\bx)]&\ge \underset{\bx\sim \mathcal{D}}{\Pr}[\bx=0^n]\cdot |C|\tag{$\depth_T(\violet{0^n})=|C|$}\\
    &= \frac{|C|}{2} \tag{$\mathcal{D}$ places weight $\frac1{2}$ on $0^n$}\\
    &\ge \frac{\opt(\mathcal{S})}{2}
\end{align*}
which contradicts our original assumption on the average depth of $T$.
\end{proof}

Our high-level proof strategy for the next claim is loosely inspired by \cite{BKLS20} (which itself built on \cite{BB19}). This proof also crucially relies on \Cref{prop:Dj_equals_D}. 
\begin{claim}[Good approximators for $\Gamma_{\oplus\ell}$ yield good approximators for $\Gamma$]
\label{claim:extracting_small_dt-new}
Let $T : (\zo^\ell)^n\to\zo$ be a decision tree such that 
\[ \dist_{\mathcal{D}_{\oplus \ell}}(T,\Gamma_{\oplus \ell})\le \eps \quad \text{and}\quad \underset{\by\sim \mathcal{D}_{\oplus \ell}}{\E}[\depth_T(\by)]\le d. \] 
Then there is a restriction $T^* : \zo^n \to \zo$ of $T$  satisfying 
\[ \dist_\mathcal{D}(T^*,\Gamma)\le 2\eps \quad \text{and} \quad 
\underset{\bx\sim \mathcal{D}}{\E}[\depth_{T^*}(\bx)]\le \frac{2d}{\ell}.
\] 
\end{claim}

\begin{proof}
 Recalling the notation from \Cref{def:alternative_Dlj},  when $\underline{z} \in (\zo^{\ell-1})^n$ and $\underline{j} \in [\ell]$ are fixed, the function $x\mapsto \ParComplete_{\underline{j}}(\underline{z},x)$ is a function from $\zo^n$ to $(\zo^{\ell})^n$. Our proof proceeds by finding a suitable $\underline{z}$ and $\underline{j}$ so that $x\mapsto T(\ParComplete_{\underline{j}}(\underline{z},x))$ is a tree of much smaller average depth and computes $\Gamma$ accurately over $\mathcal{D}$. Restricting $T$ according to the values specified by $\underline{z}$ and $\underline{j}$ yields the desired decision tree. 

For $j\in [\ell]$ and $y\in (\zo^\ell)^n$, write $q_j(y)$ for the number of times that $T$, on the input $y$, queries $(y_i)_j$ for some $i\in[n]$. Thus, $\depth_T(y)=\sum_{j\in[\ell]} q_j(y)$ and likewise
\[
\sum_{j\in[\ell]} \underset{\by\sim \mathcal{D}_{\oplus \ell}}{\E}[q_j(\by)]=\underset{\by\sim \mathcal{D}_{\oplus \ell}}{\E}[\depth_T(\by)]\le d.
\] 
Let $\underline{j}\in [\ell]$ be the index that minimizes $\underset{\by\sim \mathcal{D}_{\oplus \ell}}{\E}[q_j(\by)]$. By averaging, this $\underline{j}$ must satisfy $\underset{\by\sim \mathcal{D}_{\oplus \ell}}{\E}[q_{\underline{j}}(\by)]\le d/\ell$.  By~\Cref{prop:Dj_equals_D}, we can write
\begin{align*} 
\frac{d}{\ell} &\ge \underset{\by\sim \mathcal{D}_{\oplus \ell}}{\E}[q_{\underline{j}}(\by)] \\
&= \Ex_{\by\sim\mathcal{D}_{\oplus \ell}^{\underline{j}}}[q_{\underline{j}}(\by)] \tag{\Cref{prop:Dj_equals_D}} \\
&=  \underset{\bz\sim \mathcal{U}_{n(\ell-1)}}{\E}\left[\underset{\bx\sim \mathcal{D}}{\E}\left[q_{\underline{j}}(\ParComplete_{\underline{j}}(\bz,\bx))\right]\right]. \tag{Definition of $\mathcal{D}_{\oplus\ell}^j$}
\end{align*} 
Similarly, we also have: 
\begin{align*}
    \eps&\ge\underset{\by\sim \mathcal{D}_{\oplus \ell}}{\Pr}[T(\by)\neq \Gamma_{\oplus\ell}(\by)]\\
    &=\underset{\by\sim \mathcal{D}^{\underline{j}}_{\oplus \ell}}{\Pr}[T(\by)\neq \Gamma_{\oplus\ell}(\by)]\tag{\Cref{prop:Dj_equals_D}}\\
    &=\underset{\bz\sim \mathcal{U}_{n(\ell-1)}}{\E}\left[\Prx_{\bx\sim\mathcal{D}}\left[T(\ParComplete_{\underline{j}}(\bz,\bx))\neq \Gamma(\bx)\right]\right].\tag{Definition of $\mathcal{D}_{\oplus\ell}^j$}
\end{align*}
Applying Markov's inequality twice, we have  
\begin{align*}
     \underset{\bz\sim \mathcal{U}_{n(\ell-1)}}{\Pr}\left[\Prx_{\bx\sim\mathcal{D}}\left[T(\ParComplete_{\underline{j}}(\bz,\bx))\neq \Gamma(\bx)\right]>2\eps\right] <\frac{1}{2}\\
\text{and} \quad \underset{\bz\sim \mathcal{U}_{n(\ell-1)}}{\Pr}\left[\underset{\bx\sim \mathcal{D}}{\E}\left[q_{\underline{j}}(\ParComplete_{\underline{j}}(\bz,\bx))\right]>\frac{2d}{\ell}\right]<\frac{1}{2}.
\end{align*}

And thus by a union bound, there is some fixed $\underline{z}\in \zo^{n(\ell-1)}$ satisfying
\[ 
\Prx_{\bx\sim\mathcal{D}}\left[T(\ParComplete_{\underline{j}}(\underline{z},\bx))\neq \Gamma(\bx)\right] \le 2\eps \quad \text{and} \quad 
\underset{\bx\sim \mathcal{D}}{\E}\left[q_{\underline{j}}(\ParComplete_{\underline{j}}(\underline{z},\bx))\right]\le \frac{2d}{\ell}.
\] 
The tree $T^*$ is formed by restricting $T$ according to $\underline{z}$ and $\underline{j}$. Also, this tree $T^*$ satisfies $\depth_{T^*}(x)=q_{\underline{j}}(\ParComplete_{\underline{j}}(\underline{z},x))$ by construction. The claim then follows.
\end{proof}

To prove~\Cref{claim:avg_depth-new}, we first need a simple proposition stating that the probability a string $\by\sim\mathcal{D}_{\oplus\ell}$ matches some fixed substring decays exponentially with the length of the substring.

\begin{proposition}[$\mathcal{D}_{\oplus\ell}$ is uniform-like]
\label{prop:avg_depth}
Let $\ell \ge 2$.  For all $R\sse [\ell]$, $r\in\zo^{|R|}$, $i\in [n]$, and $b \in \zo$, we have \[ \underset{\by\sim \mathcal{D}_{\oplus\ell}}{\Pr}[(\by_i)_R=r\mid \oplus \by_i =b]\le 2^{-|R|/2}
\] 
where $(\by_i)_R\in\zo^{|R|}$ is the substring of $\by_i\in\zo^{\ell}$ consisting of the coordinates specified by $R$.
\end{proposition}

\begin{proof}
We first consider the case when $|R| < \ell$.  By the definition of~$\mathcal{D}_{\oplus \ell}$, the conditional distribution in question is the uniform distribution over all strings in $\zo^{\ell}$ whose parity is $b$.  The marginal distribution of this distribution over any set of $|R| < \ell$ coordinates is uniform, and therefore: 
\[ \Prx_{\by\sim\mathcal{D}_{\oplus \ell}}[(\by_i)_R = r \mid \oplus \by_i = b] = 2^{-|R|}.  \] 
If $|R|=\ell$, then depending on whether the parity of the bits in $r$ match $b$, we have:
\[ 
\Prx_{\by\sim\mathcal{D}_{\oplus \ell}}[\by_i = r \mid \oplus \by_i = b] = 
\begin{cases}
0 & \text{if $\oplus r \ne b$} \\
2^{-|R|+1} & \text{if $\oplus r = b$}. 
\end{cases}
\]  
In either case, we have the desired probability bound. 
\end{proof}

\begin{claim}[Small trees have small average depth]
\label{claim:avg_depth-new}
Let $T$ be a size-$s$ decision tree, then $$\underset{\by\sim \mathcal{D}_{\oplus \ell}}{\E}\left[\depth_T(\by)\right]\le 2\log s.$$
\end{claim}

\begin{proof}
 We start by upper bounding $\Pr[\text{$\by$ reaches $L$}]$ for any fixed leaf $L$ of $T$.  For each block $i\in [n]$, we write $R_i(L)$ to denote the variables from the $i$th block queried on the root-to-$L$ path, and $r_i(L)\in \zo^{R_i(L)}$ to denote the values that the path assigns to these variables.  Note that $\sum_{i\in [n]} |R_i(L)| = |L|$, the depth of $L$ in $T$.  With this notation in hand, for any fixed $x\in \zo^n$, we have 
\begin{align*} 
&\Prx_{\by\sim\mathcal{D}_{\oplus \ell}}[\text{$\by$ reaches $L$}\mid \BlockwisePar(\by) = x] \\
&=\prod_{i\in [n]} \Prx_{\by\sim\mathcal{D}_{\oplus \ell}}[(\by_i)_{R_i(L)} = r_i(L) \mid \BlockwisePar(\by) = x] \tag{Independence of the $\by_i$'s for fixed $x$} \\
&= \prod_{i\in [n]} \Prx_{\by\sim\mathcal{D}_{\oplus \ell}}[(\by_i)_{R_i(L)} = r_i(L) \mid \oplus \by_i = x_i] \\
&\le \prod_{i\in [n]} 2^{-|R_i(L)|/2} \tag{\Cref{prop:avg_depth}} \\
&= 2^{-|L|/2}. 
\end{align*} 
Since this holds for every $x$, it follows that 
\begin{equation}  \Prx_{\by\sim\mathcal{D}_{\oplus \ell}}[\text{$\by$ reaches $L$}] \le 2^{-|L|/2}.\label{eq:exponential} 
\end{equation} 
We therefore conclude that 
\begin{align*} 
\lfrac1{2}\cdot \underset{\by\sim \mathcal{D}_{\oplus \ell}}{\E}\left[\depth_T(\by)\right]&=
\Ex_{\by\sim\mathcal{D}_{\oplus \ell}}\left[\log\big(2^{\depth_T(\by)/2}\big)\right] \\
&\le \log\left(\Ex_{\by\sim\mathcal{D}_{\oplus \ell}}[2^{\depth_T(\by)/2}]\right) \tag{Concavity of $\log(\cdot)$}\\
&= \log\left( \sum_{L \in T} \Prx_{\by\sim\mathcal{D}_{\oplus \ell}}[\text{$y$ reaches $L$}] \cdot 2^{|L|/2} \right) \\ 
&\le  \log\left( \sum_{L\in T} 2^{-|L|/2} \cdot 2^{|L|/2} \right)  \tag{\Cref{eq:exponential}} \\
&= \log s.
\end{align*}
Rearranging completes the proof. 
\end{proof}

\paragraph{Putting things together: Proof of~\Cref{lem:set_cover-reduction}.} Suppose there is some tree $T$ computing $\Gamma_{\oplus \ell}$ with $|T|\le 2^{\opt(\mathcal{S})\ell/8}$. We show that $\dist(T,\Gamma_{\oplus \ell})\ge 1/(4N)$. Suppose for contradiction that $\dist(T,\Gamma_{\oplus\ell})<1/(4N)$. By \Cref{claim:avg_depth-new}, we have $\eyd{\depth_T(\by)}< 2\cdot\log\left(2^{\opt(\mathcal{S})\ell/8}\right)=\opt(\mathcal{S})\ell/4$. Then by \Cref{claim:extracting_small_dt-new} there is a decision tree $T^*$ satisfying
$$
\dist_\mathcal{D}(T^*,\Gamma)< \frac{1}{2N}\qquad \underset{\bx\sim \mathcal{D}}{\E}[\depth_{T^*}(\bx)]< \frac{\opt(\mathcal{S})}{2}.
$$
But this contradicts \Cref{claim:depth_error-new}. \hfill{$\square$} 

\subsection{Proof of~\Cref{lem:Set-Cover-to-DTs}: Lower bounds against DNF hypotheses}
We extend \Cref{lem:set_cover-reduction} to show that $\Gamma_{\oplus\ell}$ cannot even be approximated by small DNFs. This extension will allow us to complete the proof of \Cref{lem:Set-Cover-to-DTs}.
For this section, we use the negation of $\Gamma$:
$$
\overline{\Gamma}(x)=
\begin{cases}
1 & x=0^n\\
0 & x=u,u\in U
\end{cases}.
$$
Analogous to \Cref{fact:Gamma-is-small-disjunction}, any set cover $C\sse S$ yields a conjunction of $k$ literals which computes $\Gamma$ under $\mathcal{D}$.
\begin{fact}
\label{fact:NGamma-is-small-conjunction}
If $\opt(\mathcal{S})\le k$, then $\overline{\Gamma}$ is a conjunction of $k$ literals under $\mathcal{D}$.
\end{fact}
The literals in this case are the negation of the variables in the set cover $C\sse S$. We will likewise use the negation of $\Gamma_{\oplus\ell}$:
$$
\overline{\Gamma}_{\oplus\ell}(y)=\overline{\Gamma}(\BlockwisePar(y)).
$$
The analogue of \Cref{fact:Gamma_ell-is-small-junta} becomes:
\begin{fact}
\label{fact:NGamma_ell-is-small-junta}
If $\opt(\mathcal{S})\le k$ then $\overline{\Gamma}_{\oplus \ell}$ is a $k\ell$-junta (a conjunction of $k$ many parities, each over $\ell$ variables) under $\mathcal{D}_{\oplus \ell}$.
\end{fact}

Ultimately, this change allows us to prove that $\overline{\Gamma}_{\oplus\ell}$ cannot be approximated by small-size DNF formulas. If instead, one were interested in proving hardness against CNF formulas, one could work directly with the unnegated $\Gamma_{\oplus\ell}$. We find that working with DNFs is slightly less cumbersome than with CNFs which is why we focus on the negated function in this section. Specifically, we prove the following extension of \Cref{lem:set_cover-reduction}.\footnote{The lemma is indeed an ``extension'' because any size-$s$ decision tree computing $\Gamma_{\oplus\ell}$ yields a size-$s$ decision tree computing $\overline{\Gamma}_{\oplus\ell}$ simply by flipping leaf labels, and so \Cref{lem:set_cover-reduction} can equivalently be viewed as a statement about $\overline{\Gamma}_{\oplus\ell}$.}

\begin{lemma}
\label{lem:NGamma-has-no-small-DNF}
Let $\mathcal{S}=(S,U,E)$ be an $N$-vertex set cover instance and let $\ell\ge 2$. If $F:(\zo^\ell)^n\to\zo$ is a DNF of size $|F|<2^{\opt(\mathcal{S})\ell/16}$, then $\dist_{\mathcal{D}_{\oplus\ell}}(\overline{\Gamma}_{\oplus\ell},F)\ge 1/(4N)$. 
\end{lemma}

The high level proof strategy follows that of \Cref{lem:set_cover-reduction} and can be divided into the same three steps outlined in \Cref{para:high-level-idea-of-DT-gamma-lb}. The only difference is that ``average depth'' is no longer a well-defined quantity with DNF formulas. Instead, we consider ``average width'' which is a generalization of average depth suited to our purposes. 

\begin{claim}[Good approximators for $\overline{\Gamma}$ require large width]
\label{claim:small-width-F-badly-approximates-NGamma}
Let $F:\zo^n\to\zo$ be a DNF formula and $\mathcal{S}=(S,U,E)$ be an $N$-vertex set cover instance with $|S|=n$. If $\exd{\width_{F}(\bx)}<\opt(\mathcal{S})/2$, then $\dist_\mathcal{D}(F,\overline{\Gamma})\ge 1/(2N)$. 
\end{claim}

\begin{proof}
Let $F=t_1\lor\cdots\lor t_s$ be a DNF formula. If $F(0^n)=0$, then $\dist_\mathcal{D}(F,\overline{\Gamma})\ge 1/2$ since $\overline{\Gamma}(0^n)=1$. Otherwise, let $t_i$ be the smallest width term such that $t_i(0^n)=1$ so that $|t_i|=\width_F(0^n)$. Since $t_i$ accepts the all $0$s input, it is a conjunction of $|t_i|$ negated variables. Let $C\sse S$ be the set of variables in $t_i$. Since
\[ 
\frac{|t_i|}{2}
    =\prxd{\bx=0^n}\cdot\width_F(0^n) \le \exd{\width_{F}(\bx)} < \frac{\opt(\mathcal{S})}{2},
\] 
$C$ is not a set cover. Let $u\in U$ be some vertex not covered by $C$: $\mathcal{N}_\mathcal{S}(u)\cap C=\varnothing$. Then, $u$ is encoded with $0$s for all variables in $C$. It follows that $t_i(u)=1$ and $F(u)=1\neq 0=\overline{\Gamma}(u)$. Therefore:
\[ 
    \dist_\mathcal{D}(F,\overline{\Gamma}) \ge \underset{\bx\sim\mathcal{D}}{\Pr}\left[{\bx=u}\right]
    =\frac{1}{2|U|}
    \ge \frac{1}{2N}. \qedhere \] 
\end{proof}

\begin{claim}[Good approximators for $\overline{\Gamma}_{\oplus\ell}$ yield good approximators for $\overline{\Gamma}$]
\label{claim:NGammaL-approximators-yield-NGamma-approximators}
Let $F:(\zo^\ell)^n\to\zo$ be a DNF formula such that
$$
\dist_{\mathcal{D}_{\oplus\ell}}(F,\overline{\Gamma}_{\oplus\ell})\le \eps\quad\text{and}\quad \eyd{\width_F(\by)}\le w.
$$
Then there is a restriction $F^*:\zo^n\to\zo$ of $F$ satisfying
$$
\dist_\mathcal{D}(F^*,\overline{\Gamma})\le 2\eps\quad\text{and}\quad \exd{\width_{F^*}(\bx)}\le \frac{2w}{\ell}.
$$
\end{claim}

\begin{proof}
The proof is similar to that of \Cref{claim:extracting_small_dt-new}. First, let $q_j(y)$ denote the number of variables of the form $(y_i)_j$ for some $i\in[n]$ appearing in the smallest width term that accepts $y$ and $0$ if no term accepts $y$. Then, $\width_F(y)=\sum_{j\in[\ell]}q_j(y)$ for all $y\in\supp(D_{\oplus\ell})$. Therefore:
$$
\sum_{j\in[\ell]} \eyd{q_j(\by)}\le w.
$$
Let $\underline{j}\in[\ell]$ be the index that minimizes $\eyd{q_j(\by)}$. By averaging, $\underline{j}$ satisfies $\eyd{q_{\underline{j}}(\by)}\le w/\ell$. Using \Cref{prop:Dj_equals_D}:
\begin{align*}
    \frac{w}{\ell}&\ge \eyd{q_{\underline{j}}(\by)}\\
    &=\underset{\by\sim\mathcal{D}_{\oplus\ell}^{\underline{j}}}{\E}\left[{q_{\underline{j}}(\by)}\right]\tag{\Cref{prop:Dj_equals_D}}\\
    &=\underset{\bz\sim\mathcal{U}_{n(\ell-1)}}{\E}\left[\exd{q_{\underline{j}}(\ParComplete_{\underline{j}}(\bz,\bx))}\right].\tag{Definition of $\mathcal{D}_{\oplus\ell}^{\underline{j}}$}
\end{align*}
Similarly:
\begin{align*}
    \eps&\ge\underset{\by\sim \mathcal{D}_{\oplus \ell}}{\Pr}[F(\by)\neq \overline{\Gamma}_{\oplus\ell}(\by)]\\
    &=\underset{\by\sim \mathcal{D}^{\underline{j}}_{\oplus \ell}}{\Pr}[F(\by)\neq \overline{\Gamma}_{\oplus\ell}(\by)]\tag{\Cref{prop:Dj_equals_D}}\\
    &=\underset{\bz\sim \mathcal{U}_{n(\ell-1)}}{\E}\left[\Prx_{\bx\sim\mathcal{D}}\left[F(\ParComplete_{\underline{j}}(\bz,\bx))\neq \overline{\Gamma}(\bx)\right]\right].\tag{Definition of $\mathcal{D}_{\oplus\ell}^j$}
\end{align*}
Applying Markov's inequality twice, we have  
\begin{align*}
 \underset{\bz\sim \mathcal{U}_{n(\ell-1)}}{\Pr}\left[\Prx_{\bx\sim\mathcal{D}}\left[F(\ParComplete_{\underline{j}}(\bz,\bx))\neq \overline{\Gamma}(\bx)\right]>2\eps\right] <\frac{1}{2}\\
 \text{and} \quad \underset{\bz\sim \mathcal{U}_{n(\ell-1)}}{\Pr}\left[\underset{\bx\sim \mathcal{D}}{\E}\left[q_{\underline{j}}(\ParComplete_{\underline{j}}(\bz,\bx))\right]>\frac{2w}{\ell}\right]<\frac{1}{2}.
\end{align*}
And thus by a union bound, there is some fixed $\underline{z}\in \zo^{n(\ell-1)}$ satisfying
\[ 
\Prx_{\bx\sim\mathcal{D}}\left[F(\ParComplete_{\underline{j}}(\underline{z},\bx))\neq \overline{\Gamma}(\bx)\right] \le 2\eps \quad \text{and} \quad 
\underset{\bx\sim \mathcal{D}}{\E}\left[q_{\underline{j}}(\ParComplete_{\underline{j}}(\underline{z},\bx))\right]\le \frac{2w}{\ell}.
\] 
The DNF formula $F^*$ is formed by restricting $F$ according to the string $\underline{z}$. Also, this $F^*$ satisfies $\width_{F^*}(x)=q_{\underline{j}}(\ParComplete_{\underline{j}}(\underline{z},x))$ by construction. The claim then follows.
\end{proof}

\begin{claim}[Small DNFs have small average width]
\label{claim:small-size-DNFs-have-small-avg-width}
Let $F$ be a size-$s$ DNF formula for $s\ge 4$ such that $\dist_{\mathcal{D}_{\oplus\ell}}(F,\overline{\Gamma}_{\oplus\ell})\le 1/4$, then
$$
\eyd{\width_F(\by)}\le 4\log (s).
$$
\end{claim}

\begin{proof}
Let $F=t_1\lor\cdots\lor t_s$ be a DNF formula with $s$ terms satisfying $\dist_{\mathcal{D}_{\oplus\ell}}(F,\overline{\Gamma}_{\oplus\ell})\le 1/4$. We start by upper bounding the conditional probability $\Pr[t(\by)=1\mid F(\by)=1]$ for any fixed term $t\in\{t_1,\ldots,t_s\}$. We bound the probabilities $\Pr[t(\by)=1]$ and $\Pr[F(\by)=1]$ separately.
\begin{enumerate}[label=(\arabic*)]
\item {$\Pr[F(\by)=1]\ge 1/4$. We write
\begin{align*}
    \frac{1}{4}&\ge \dist_{\mathcal{D}_{\oplus\ell}}(F,\overline{\Gamma}_{\oplus\ell})\\
    &\ge \left|\pryd{F(\by)=1}-\pryd{\overline{\Gamma}_{\oplus\ell}(\by)=1}\right|\\
    &=\left|\pryd{F(\by)=1}-\frac{1}{2}\right|
\end{align*}
which implies $\Pr[F(\by)=1]\ge 1/4$.
}
\item {$\Pr[t(\by)=1]\le 2^{-|t|/2}$. For each $i\in [n]$, let $R_i(t)$ denote the variables from the $i$th block which appear in the term $t$ and let $r_i(t)\in\zo^{R_i(t)}$ denote the values assigned by those variables (i.e. $1$ if the variable is unnegated in $t$ and $0$ if the variable is negated in $t$). Then $\sum_{i\in[n]}|R_i(t)|=|t|$, the width of $t$. Using this notation, for any fixed $x\in\supp(\mathcal{D})$:
\begin{align*}
    &\pryd{t(\by)=1\mid \BlockwisePar(\by)=x}\\
    &=\prod_{i\in[n]}\pryd{(\by_i)_{R_i(t)}=r_i(t)\mid\BlockwisePar(\by)=x}\tag{Independence of the $\by_i$'s for fixed $x$}\\
    &=\prod_{i\in[n]}\pryd{(\by_i)_{R_i(t)}=r_i(t)\mid\oplus \by_i=x_i}\\
    &\le \prod_{i\in[n]}2^{-|R_i(t)|/2}\tag{\Cref{prop:avg_depth}}\\
    &=2^{-|t|/2}.
\end{align*}
Since this holds for any fixed $x$, it follows that
$$
\pryd{t(\by)=1}\le 2^{-|t|/2}.
$$
}
\end{enumerate}
Together, these two points imply
\begin{equation}
    \label{eq:pr-t(y)=1-given-1-is-small}
    \pryd{t(\by)=1\mid F(\by)=1}=\frac{\Pr[t(\by)=1]}{\Pr[F(\by)=1]}\le 2^{-|t|/2+2}.
\end{equation}
Lastly:
\begin{align*}
    \frac{1}{2}\cdot\eyd{\width_F(\by)}-2&=\eyd{\log\paren{2^{\width_F(\by)/2-2}}}\\
    &\le \log\paren{\eyd{2^{\width_F(\by)/2-2}}}\tag{Concavity of $\log$}\\
    &=\log\paren{\sum_{b\in\zo}\pryd{F(\by)=b}\cdot\eyd{2^{\width_F(\by)/2-2}\mid F(\by)= b}}\\
    &\le \log\paren{\eyd{2^{\width_F(\by)/2-2}\mid F(\by)= 1}}\tag{$\width_F(\by)=0$ if $F(\by)=0$ and $\Pr[F(\by)=b]\le 1$}\\
    &\le \log\paren{\sum_{i\in[s]}2^{|t_i|/2}\cdot \pryd{t_i(\by)=1\mid F(\by)=1}}\\
    &= \log\paren{\sum_{i\in[s]}2^{|t_i|/2-2}\cdot 2^{-|t_i|/2+2}}\tag{\Cref{eq:pr-t(y)=1-given-1-is-small}}\\
    &=\log s.
\end{align*}
Rearranging and applying the assumption that $2\le \log(s)$ completes the proof.
\end{proof}

\paragraph{Putting things together: Proof of~\Cref{lem:NGamma-has-no-small-DNF}} Suppose there is some DNF formula $F$ computing $\overline{\Gamma}_{\oplus \ell}$ with $|F|\le 2^{\opt(\mathcal{S})\ell/16}$. We show that $\dist(F,\overline{\Gamma}_{\oplus \ell})\ge 1/(4N)$. Suppose for contradiction that $\dist(T,\overline{\Gamma}_{\oplus\ell})<1/(4N)\le 1/4$. If $|F|<4$, we add dummy terms (e.g. by replicating the terms already in $F$) so that $|F|\ge 4$. We can then apply \Cref{claim:small-size-DNFs-have-small-avg-width}: $\eyd{\width_F(\by)}< 4\cdot\log\left(2^{\opt(\mathcal{S})\ell/16}\right)=\opt(\mathcal{S})\ell/4$. Then by \Cref{claim:NGammaL-approximators-yield-NGamma-approximators}, there is a DNF formula $F^*$ satisfying
$$
\dist_\mathcal{D}(F^*,\Gamma)< \frac{1}{2N}\qquad \underset{\bx\sim \mathcal{D}}{\E}[\depth_{F^*}(\bx)]< \frac{\opt(\mathcal{S})}{2}.
$$
But such an $F^*$ contradicts \Cref{claim:small-width-F-badly-approximates-NGamma}. \hfill{$\square$}

\paragraph{The last steps: finishing the proof of \Cref{lem:Set-Cover-to-DTs}.}{
We prove the following lemma which immediately implies \Cref{lem:Set-Cover-to-DTs}.
\begin{lemma}[$\overline{\Gamma}_{\oplus\ell}$ proves \Cref{lem:Set-Cover-to-DTs}]
\label{lem:given-set-cover-construct-NGamma-which-is-hard-to-approximate-by-DNFs}
Let $\mathcal{S}=(S,U,E)$ be an $N$-vertex instance of $(k,k')$-{\sc Set-Cover} and $\ell\le N$. Then there is an algorithm that runs in $\poly(N)$ time and outputs a circuit representation of $\overline{\Gamma}_{\oplus\ell}$ under $\mathcal{D}_{\oplus\ell}$ and a generator for $\mathcal{D}_{\oplus\ell}$ which satisfies:
\begin{itemize}
    \item[$\circ$] If $\opt(\mathcal{S})\le k$, then $\overline{\Gamma}_{\oplus\ell}$ is a $k\ell$-junta under $\mathcal{D}_{\oplus\ell}$.
    \item[$\circ$] If $\opt(\mathcal{S})> k'$, then any DNF of size $\le 2^{k'\ell/16}$ is $\frac{1}{4N}$-far from $\overline{\Gamma}_{\oplus\ell}$ under $\mathcal{D}_{\oplus\ell}$.
\end{itemize}
\end{lemma}
\begin{proof}
By \Cref{lem:circuit-for-gamma-generator-for-D}, there is an algorithm that runs in $\poly(N)$ time and outputs a circuit representation of ${\Gamma}_{\oplus\ell}$ and a generator for $\mathcal{D}_{\oplus\ell}$. Augmenting the circuit for ${\Gamma}_{\oplus\ell}$ with a single NOT gate yields a circuit for $\overline{\Gamma}_{\oplus\ell}$. Moreover, we have shown:
\begin{itemize}[leftmargin=0.8cm]
    \item[$\circ$] if $\opt(\mathcal{S})\le k$, then $\overline{\Gamma}_{\oplus\ell}$ is a $k\ell$-junta under $\mathcal{D}_{\oplus\ell}$;\hfill{(\Cref{fact:NGamma_ell-is-small-junta})}
    \item[$\circ$] if $\opt(\mathcal{S})>k'$, then any DNF of size $\le 2^{k'\ell/16}$ is $\frac{1}{4N}$-far from $\overline{\Gamma}_{\oplus\ell}$ under $\mathcal{D}_{\oplus\ell}$;\hfill{(\Cref{lem:NGamma-has-no-small-DNF})} 
\end{itemize}
which completes the proof of the lemma.
\end{proof}
}

\subsection{Implications of~\Cref{lem:Set-Cover-to-DTs}}

\subsubsection{Proofs of~\Cref{thm:DT-DT} and~\Cref{thm:junta-DNF}}
In this section, we prove the following theorem.
\begin{theorem}
\label{thm:DNF-construction-hardness}
Let $\mu:\N\to\N$ be any computable, non-decreasing function satisfying $\mu(n)=o\paren{\frac{\log\log n}{\log\log\log n}}$. Assuming randomized ETH, there is some constant $\lambda\in (0,1)$, a function $f:\zo^{n}\to\zo$, and distribution $\mathcal{D}$ over $\zo^n$ such that ~{\sc DT-Construction$(s,1/n)$} cannot be solved in time $${s^{\lambda\cdot\left(\frac{\log\log s}{\mu(n)\log\log\log s}\right)}}$$ for $f$ and for any $s\le n^{\mu(n)}$, even if $f$ is promised to be a $(\log n)$-junta over $\mathcal{D}$ and the algorithm returns a DNF hypothesis. 
\end{theorem}
\Cref{thm:DT-DT,thm:junta-DNF} immediately follow as a consequence of this theorem by choosing $\mu(n)=1$.

\begin{proof}[Proof of~\Cref{thm:DNF-construction-hardness}]
We give a reduction from gapped set cover. Let $\mathcal{S}=(S,U,E)$ be an $N$-vertex $\left(k,\frac{1}{2}\left(\frac{\log N}{\log\log N}\right)^{1/k}\right)$-\setcover instance where $k$ is taken to be
$$
k=\frac{1}{2}\cdot\frac{\log\log N}{\log\log\log N}.
$$
Using \Cref{lem:given-set-cover-construct-NGamma-which-is-hard-to-approximate-by-DNFs} with $\ell=\log(N)/k$, we obtain the target function $\overline{\Gamma}_{\oplus \ell}:\zo^{N\ell}\to\zo$ and the distribution $\mathcal{D}_{\oplus\ell}$.\footnote{Technically, $\overline{\Gamma}_{\oplus \ell}$ is a function defined on $|S|\ell$ bits, but as $|S|\le N$ we can pad the inputs to be $N\ell$ bits long.} 

Let $\mu:\N\to\N$ be as in the theorem statement. Set $s\coloneqq (N\ell)^{\mu(N\ell)}$. We show that any algorithm for {\sc DT-Construction$(s,1/(4N))$} running in time $\ds {s^{\lambda\cdot\left(\frac{\log\log s}{\mu(N\ell)\log\log\log s}\right)}}$ for $0<\lambda\le 1/128$ can be used to solve $\mathcal{S}$ in time $N^{8\lambda\cdot k}$ even if the output of the algorithm is a DNF formula. 

We run the algorithm for {\sc DT-Construction$(s,1/(4N))$} on $\overline{\Gamma}_{\oplus \ell}$ and $\mathcal{D}_{\oplus\ell}$ and terminate it after
$$
N^{4\lambda\cdot \left(\frac{\log\log N}{\log\log\log N}\right)}=N^{8\lambda k}
$$
times steps. The algorithm outputs some DNF formula $F$. We estimate the error of $F$ and $\overline{\Gamma}_{\oplus \ell}$ over the distribution $\mathcal{D}_{\oplus\ell}$ and output ``\textsc{Yes}'' if the error is $\le 1/(4N)$ and ``\textsc{No}'' otherwise.

\paragraph{Runtime.}{
Constructing the circuit for $\overline{\Gamma}_{\oplus \ell}$ and the generator for $\mathcal{D}_{\oplus\ell}$ requires $\poly(N)$ time by \Cref{lem:given-set-cover-construct-NGamma-which-is-hard-to-approximate-by-DNFs}.
We can efficiently sample from the distribution $\mathcal{D}_{\oplus\ell}$ to efficiently estimate the error of the output decision tree via random samples. So the overall runtime of our algorithm is $\le N^{8\lambda k}$. 
}

\paragraph{Correctness.}{
To prove the reduction is correct, we show that if there is a size $k$ set cover for $\mathcal{S}$ then we output \textsc{Yes} with high probability and otherwise if $S$ requires a set cover of size at least
$$
\frac{1}{2}\left(\frac{\log N}{\log\log N}\right)^{1/k}
$$
then we output \textsc{No} with high probability. \bigskip

{\sc Yes} {\bf case: $\opt(S)\le k$.} {In this case, by \Cref{lem:given-set-cover-construct-NGamma-which-is-hard-to-approximate-by-DNFs}, $\overline{\Gamma}_{\oplus\ell}$ is computed exactly by a $\opt(S)\ell\le k\ell=\log N$-junta over $\mathcal{D}_{\oplus\ell}$. Hence, it is computed by a DNF of width $k\ell$. The size of this DNF is at most $2^{k\cdot\ell}=N\le (N\ell)^{\mu(N\ell)}=s$. To upper bound the runtime, we start by calculating
\begin{align}
    \frac{\log\log s}{\log\log\log s}&\le \frac{\log\paren{2\mu(N^2)\log N}}{\log\log\log N}\tag{$N\le s\le N^{2\mu(N^2)}$}\nonumber\\
    &\le \frac{\log\paren{(\log N)^2}}{\log\log\log N}\tag{Assumption on $\mu$: $2\mu(N^2)\le \log N$}\\
    &=4k. \label{eq:lglgs-power-upper-bound}
\end{align}
By our assumption on {\sc DT-Construction$(s,1/(4N))$}, in the yes case, the algorithm runs for
\begin{align*}
    s^{\lambda\cdot\left(\frac{\log\log s}{\mu(N\ell)\log\log\log s}\right)}&\le s^{4\lambda k/\mu(N\ell)}\tag{\Cref{eq:lglgs-power-upper-bound}}\\
    &=(N\ell)^{4\lambda k}\tag{$s=(N\ell)^{\mu(N\ell)}$}\\
    &\le N^{8\lambda k}\tag{$N\ell\le N^2$}
\end{align*}
time steps and outputs a size-$s$ DNF formula with error $\le 1/(4N)$. Therefore, our algorithm outputs \textsc{Yes} with high probability (where the probability is taken over the random sampling procedure). 
}\bigskip 

{\sc No} {\bf case: $\opt(S)>\frac{1}{2}\left(\frac{\log N}{\log\log N}\right)^{1/k}.$} {By \Cref{lem:given-set-cover-construct-NGamma-which-is-hard-to-approximate-by-DNFs} any DNF for $\overline{\Gamma}_{\oplus\ell}$ with size at most $2^{\opt(\mathcal{S})\ell/16}$ has error at least $1/(4N)$. The runtime bound on our algorithm serves as an upper bound on the size of the DNF built by the {\sc DT-Construction} algorithm. Therefore, it is sufficient to show that 
\begin{equation}
    \label{eq:to_prove_size_bound}
    N^{8\lambda \cdot k}<2^{\opt(S)\ell/16}
\end{equation}
because this bound shows that our DNF must have error at least $1/(4N)$.  
Recalling that $k=\frac{1}{2}\cdot\frac{\log\log N}{\log\log\log N}$, we have $(2k^2)^k<\frac{\log N}{\log\log N}$. We observe
\begin{align*}
    \opt(\mathcal{S})&>\frac{1}{2}\left(\frac{\log N}{\log\log N}\right)^{1/k}\\
    &>k^2\\
    &\ge 128\lambda k^2\tag{$128\lambda\le 1$}
\end{align*}
which shows $k\ell(8\lambda k)<\opt(\mathcal{S})\ell/16$. Exponentiating both sides and using the fact that $N=2^{k\ell}$ completes the calculation and establishes \Cref{eq:to_prove_size_bound}. It follows that our algorithm finds the error to be $>1/(4N)$ and outputs \textsc{No} with high probability.
}
}
\paragraph{Refuting randomized ETH.}{
We now have an algorithm for solving $\left(k,\frac{1}{2}\left(\frac{\log N}{\log\log N}\right)^{1/k}\right)$-\setcover in time $N^{8\lambda k}$ with high probability. By \Cref{thm:SC_hardness}, there is a constant $c\in(0,1)$ such that $\left(k,\frac{1}{2}\left(\frac{\log N}{\log\log N}\right)^{1/k}\right)$-\setcover cannot be solved with high probability in time $N^{ck}$. Therefore, we derive a contradiction for any $\lambda\le \min\{c/8,1/128\}$.
}
\end{proof}

\subsubsection{PAC learning hardness}
\label{subsub:PAC-learning-hardness}
{
In this section, we discuss corollaries of \Cref{thm:DNF-construction-hardness}. For a brief background on PAC learning and the definitions that we use, see \Cref{appendix:pac-learning}. 

\begin{corollary}[Hardness of learning decision trees, DNFs, and CNFs]
\label{cor:pac-learning-hardness-of-dts-dnfs}
Assuming randomized ETH, there is a constant $\lambda\in(0,1)$ such that decision trees cannot be distribution-free, properly PAC learned to accuracy $\eps=1/n$ in time $s^{\lambda\frac{\log\log s}{\log\log\log s}}$ where $s$ is the size of the target. The same result also holds for properly learning DNFs and CNFs with size-$s$ targets.
\end{corollary}

\begin{proof}
Let $\mathcal{L}$ be a distribution-free, proper learning algorithm for the class $\mathcal{T}$ of decision trees. We claim $\mathcal{L}$ can be used to solve {\sc DT-Construction}. In particular, let $f:\zo^n\to\zo$ and $\mathcal{D}$ be an instance of {\sc DT-Construction}$(s,1/n)$. We run the learning algorithm on $f$ and $\mathcal{D}$ and $\eps=1/n$. If $\mathcal{L}$ requests a random sample, we generate $x\sim\mathcal{D}$ using the generator for $\mathcal{D}$ and evaluate $f(x)$ using the circuit for $f$ and return $(x,f(x))$ to $\mathcal{L}$. Since generating a sample from $\mathcal{D}$ and evaluating the circuit for $f$ are both $\poly(n)$-time operations the overall runtime is dominated by the runtime of $\mathcal{L}$. \Cref{thm:DNF-construction-hardness} then implies the desired time bound by setting $\mu(n)=1$. 

If $\mathcal{L}$ is a learning algorithm for DNFs, we obtain the same hardness as in the decision tree case since any size-$s$ decision tree target is equivalently a size-$s$ DNF target. Moreover, \Cref{thm:DNF-construction-hardness} also applies when the output of the {\sc DT-Construction} algorithm is a DNF formula. A symmetric argument works similarly for CNFs.
\end{proof}

}

\subsubsection{Proof of~\Cref{thm:junta-junta}}
In this section, we observe that the number of relevant inputs to $\Gamma_{\mathcal{S}}$ exactly characterizes the set cover complexity of $\mathcal{S}$. As a result, hardness of approximating set cover can be directly translated into hardness of distribution-free, proper PAC learning $k$-juntas.  The next theorem formalizes this observation and was already implicit in \cite{ABFKP09}.
\begin{theorem}[Learning $k$-juntas is as hard as {\sc Set-Cover}]
\label{thm:learning-k-juntas-hard-as-set-cover}
Suppose there is a distribution-free PAC learning algorithm that runs in time $t(n,k)$ and learns the class of $k$-juntas over $\zo^n$ to accuracy $\eps=O(1/n)$ by hypotheses which are $g(k,n)$-juntas for some function $g:\N^2\to \N$ satisfying $k\le g(k,n)$. Then $(k,g(k,n))$-{\sc Set-Cover} can be solved with high probability in time $t(n,k)$.
\end{theorem}

\begin{proof}
Let $\mathcal{S}=(S,U,E)$ be an instance of $(k,g(k,n))$-{\sc Set-Cover}. We construct the function $\Gamma:\zo^{|S|}\to\zo$ and the distribution $\mathcal{D}$ over $\zo^{|S|}$. Run the learning algorithm on $\Gamma$ and $\mathcal{D}$ with $\eps=1/(4|S|)$ for $t(k,|S|)$ time steps. It outputs some truth table representation of a junta. We output \textsc{Yes} if and only if this truth table has size at most $g(k,n)$ and has error at most $1/(4|S|)$. The correctness of the reduction follows from \Cref{fact:Gamma-is-small-disjunction}.
\end{proof}

\begin{corollary}
There is no distribution-free PAC learning algorithm for properly learning $k$-juntas to accuracy $\eps=O(1/n)$ over $\zo^n$ that runs in time:
\begin{itemize}
    \item[$\circ$] $n^{o(k)}$, assuming randomized $W[1]\neq \textnormal{FPT}$;
    \item[$\circ$] $O(n^{k-\lambda})$, for all $\lambda>0$, assuming randomized SETH. 
\end{itemize}
These results hold in the regime where $k\le n^c$ for some absolute constant $c<1$.
\end{corollary}

\begin{proof}
By \Cref{thm:learning-k-juntas-hard-as-set-cover}, distribution-free properly PAC learning $k$-juntas is equivalent to $(k,k+1)$-{\sc Set-Cover}. The first bullet follows by combining \Cref{thm:learning-k-juntas-hard-as-set-cover,thm:FPT-hardness-of-ungapped-set-cover}. The second bullet follows by combining \Cref{thm:learning-k-juntas-hard-as-set-cover,thm:SETH-hardness-of-ungapped-set-cover}.
\end{proof}

\section{Lower bounds for {\sc DT-Estimation}} 
\label{section:lower_bounds_dt_estimation}

For our lower bounds for {\sc DT-Estimation}, we  have to consider decision trees that are allowed to {\sl abort}: 

\begin{definition} 
A $\delta$-abort decision tree $T$ under a distribution $\mathcal{D}$ is a decision tree with leaves labeled $\{0,1,\bot\}$ satisfying $\Prx_{\bx \sim \mathcal{D}}[T(\bx) = \bot] \le \delta$.   The distance between such a tree $T : \zo^n \to \{0,1,\bot\}$ and a function $f : \zo^n \to \zo$ under $\mathcal{D}$ is
\[ \dist_\mathcal{D}(T,f) \coloneqq \Prx_{\bx\sim\mathcal{D}}[T(\bx) \ne f(\bx) \text{ and $T(\bx) \ne \bot$}]. \] 
\end{definition}

\subsection{\Cref{lem:set_cover-reduction} for decision trees that abort} 

In this section we generalize~\Cref{lem:set_cover-reduction} to $\delta$-abort decision trees: 

\begin{lemma}
[\Cref{lem:set_cover-reduction} with aborts]
\label{lem:set_cover-reduction-with-aborts} 
Let $\mathcal{S}=(S,U,E)$ be an $N$-vertex set cover instance and let $\ell\in \N$.  If $T : (\zo^\ell)^n \to \zo$ is a decision tree of size $|T|< 2^{\opt(\mathcal{S}) {\ell}/40}$ that can abort with probability $\delta < 0.4$, then  $\dist_{\mathcal{D}_{\oplus\ell}}(T,\Gamma_{\oplus\ell})\ge 1/(20N)$.
\end{lemma}

Since every depth-$d$ tree is a tree of size $2^d$,~\Cref{lem:set_cover-reduction-with-aborts} also holds for decision trees $T$ of depth $< \opt(\mathcal{S})\ell/40$. 

\paragraph{Outline of Proof.} As in the non-abort case, there are three main components to the proof of the above lemma: 

\begin{enumerate}[itemsep=0.1em]
    \item No $\delta$-abort decision tree with small average depth can approximate $\Gamma$ under $\mathcal{D}$ where $\delta < 1/2$ (\Cref{claim:depth_error-new-abort}).
    \item Any $\delta$-abort decision tree with small average depth that approximates $\Gamma_{\oplus \ell}$ under $\mathcal{D}_{\oplus \ell}$ can be used to construct a $\delta$-abort decision tree of much smaller average depth that approximates $\Gamma$ under $\mathcal{D}$ at the cost of a modest blowup in the size of $\delta$ (\Cref{claim:extracting_small_dt-new-abort}). This is the key claim. 
    \item Any small size decision tree must have small average depth with respect to $\mathcal{D}_{\oplus \ell}$. This claim is unchanged from the non-abort version. 
\end{enumerate}

Analogous to the non-abort case, these claims together imply that no $\delta$-abort decision tree with small average depth can approximate $\Gamma_{\oplus \ell}$ under $\mathcal{D}_{\oplus \ell}$. We need to provide slightly different claims and proofs for the first two items, but the last claim is completely independent of aborts, so we need not reprove it. 

\begin{claim}[Abort version of 
\Cref{claim:depth_error-new}]
\label{claim:depth_error-new-abort}
Let $T : \zo^n \to \zo$ be a $\delta$-abort decision tree with $\delta<1/2$ and $\mathcal{S}=(S,U,E)$ be an $N$-vertex set cover instance with $|S|=n$. If $\underset{\bx\sim \mathcal{D}}{\E}[\depth_T(\bx)]<\opt(\mathcal{S})/2$ then $\dist_{\mathcal{D}}(T,\Gamma)\ge 1/(2N)$.
\end{claim}

\begin{proof}
This proof is almost identical to that of~\Cref{claim:depth_error-new}. We provide the start of the proof and then refer the reader back to~\Cref{claim:depth_error-new} for the rest. 

Suppose that $\dist_{\mathcal{D}}(T,\Gamma)< 1/(2N) \leq 1/(2|U|)$. We note as before that each $x \in \text{supp}(\mathcal{D})$ has mass $\geq 1/(2|U|)$ under $\mathcal{D}$, so it must be that $\dist_{\mathcal{D}}(T,\Gamma)=0$. Since $0^n$ has weight $1/2$ under $\mathcal{D}$, $T(0^n) \neq \perp$ because $T$ can only abort with probability $<1/2$. It follows that $T(0^n) = \Gamma(0^n) = 0$. 

The rest of the proof is identical to that of~\Cref{claim:depth_error-new}
\end{proof}

For the next claim, we reuse the portions of~\Cref{claim:extracting_small_dt-new} that tell us that the restriction $T^*$ of $T$ is distance preserving and has small depth. We must show that $T^*$ also has a small abort probability.  

\begin{claim}[Abort version of \Cref{claim:extracting_small_dt-new}]
\label{claim:extracting_small_dt-new-abort}
Let $T : (\zo^\ell)^n\to\zo$ be a decision tree such that 
\[ \dist_{\mathcal{D}_{\oplus \ell}}(T,\Gamma_{\oplus \ell})\le \eps \quad \text{and}\quad \underset{\by\sim \mathcal{D}_{\oplus \ell}}{\E}[\depth_T(\by)]\le d \quad \text{and} \quad \Prx_{\by \sim \mathcal{D}_{\oplus \ell}}[T(\by) = \bot] \le \delta. \] 
Then there is a restriction $T^* : \zo^n \to \zo$ of $T$  satisfying 
\[ \dist_\mathcal{D}(T^*,\Gamma)\le 10\eps \quad \text{and} \quad 
\underset{\bx\sim \mathcal{D}}{\E}[\depth_{T^*}(\bx)]\le \frac{10d}{\ell} \quad \text{and} \quad \Prx_{\bx \sim \mathcal{D}}[T^*(\bx) = \bot] \le \frac{5}{4}\delta. \] 
\end{claim}

\begin{proof}
Recall the definition of $q_j(y)$ in the proof of~\Cref{claim:extracting_small_dt-new}. For $j\in [\ell]$, $q_j(y)$ is the number of times that $T$, on input $y$, queries $(y_i)_j$ for some $i\in[n]$. Refer to \Cref{def:alternative_Dlj} for the definitions of $\mathcal{D}^{j}_{\oplus \ell}$ and $\ParComplete_{j}(z,x)$. The proof of \Cref{claim:extracting_small_dt-new} bounds the probabilities that $\dist_\mathcal{D}(T^*,\Gamma)$ or  $ 
\underset{\bx\sim \mathcal{D}}{\E}[\depth_{T^*}(\bx)]$ are too large using Markov's inequality. More concretely, we already know that for a particular $\underline{j}\in [\ell]$, 
\begin{align*}
  \underset{\bz\sim \mathcal{U}_{n(\ell-1)}}{\Pr}\left[\Prx_{\bx\sim\mathcal{D}}\left[ T(\ParComplete_{\underline{j}}(\bz,\bx))\neq \Gamma(\bx) \text{ and } T(\ParComplete_{\underline{j}}(\bz,\bx)) \neq \perp \right]>10\eps\right] <\frac{1}{10}\\
\text{and} \quad \underset{\bz\sim \mathcal{U}_{n(\ell-1)}}{\Pr}\left[\underset{\bx\sim \mathcal{D}}{\E}\left[q_{\underline{j}}(\ParComplete_{\underline{j}}(\bz,\bx))\right]>\frac{10d}{\ell}\right]<\frac{1}{10}
\end{align*}    
where all we have done is change the constant used in the application of Markov's inequality. 

It remains for us to bound the probability that the tree aborts. We compute
\begin{align*}
    \delta&\ge\underset{\by\sim \mathcal{D}_{\oplus \ell}}{\Pr}[T(\by) = \perp]\\
    &=\underset{\by\sim \mathcal{D}^{\underline{j}}_{\oplus \ell}}{\Pr}[T(\by)= \perp]\tag{\Cref{prop:Dj_equals_D}}\\
    &=\underset{\bz\sim \mathcal{U}_{n(\ell-1)}}{\E}\left[\Prx_{\bx\sim\mathcal{D}}\left[T(\ParComplete_{\underline{j}}(\bz,\bx)) = \perp \right]\right].\tag{Definition of $\mathcal{D}_{\oplus\ell}^j$}
\end{align*}

Again, we can apply Markov's inequality to deduce 
\begin{align*}
    \underset{\bz\sim \mathcal{U}_{n(\ell-1)}}{\Pr}\left[\Prx_{\bx\sim\mathcal{D}}\left[ T(\ParComplete_{\underline{j}}(\bz,\bx)) = \perp \right]>\frac{10}{8}\delta\right] <\frac{8}{10}.
\end{align*}
Thus, applying a union bound to all three of our Markov inequalities, we conclude that there exists a fixed $\underline{z} \in \{0,1\}^{n(\ell-1)}$ such that
\begin{align*}
\Prx_{\bx\sim\mathcal{D}}\left[T(\ParComplete_{\underline{j}}(\underline{z},\bx))\neq \Gamma(\bx) \text{ and } T(\ParComplete_{\underline{j}}(\underline{z},\bx)) \neq \perp\right] \le 10\eps\\
\text{and} \quad 
\underset{\bx\sim \mathcal{D}}{\E}\left[q_{\underline{j}}(\ParComplete_{\underline{j}}(\underline{z},\bx))\right]\le \frac{10d}{\ell}\\
\text{and} \quad 
\Prx_{\bx\sim\mathcal{D}}\left[ T(\ParComplete_{\underline{j}}(\underline{z},\bx)) = \perp \right]\leq \frac{10}{8}\delta
\end{align*}
The tree $T^*$ is formed by restricting $T$ according to $\underline{z}$ and $\underline{j}$. As before, the depth of an input $x$ is $\text{depth}_{T^*}(x) = q_{\underline{j}}(\ParComplete_{\underline{j}}(\underline{z},x))$. Thus, the claim follows. 
\end{proof}

Finally, we can directly apply ~\Cref{claim:avg_depth-new} without needing a special version for aborts. These three claims together allow us to complete the proof. 

\paragraph{Putting things together: Proof of~\Cref{lem:set_cover-reduction-with-aborts}.} 
Suppose there is some $\delta$-abort tree $T$ computing $\Gamma_{\oplus \ell}$ with $|T|\le 2^{\opt(\mathcal{S})\ell/40}$ and with $\delta < 0.4$. We show that $\dist(T,\Gamma_{\oplus \ell})\ge 1/(20N)$. Suppose for contradiction that $\dist(T,\Gamma_{\oplus\ell})<1/(20N)$. By \Cref{claim:avg_depth-new}, we have $\eyd{\depth_T(\by)}< 2\cdot\log\left(2^{\opt(\mathcal{S})\ell/40}\right)=\opt(\mathcal{S})\ell/20$. Then by \Cref{claim:extracting_small_dt-new-abort} there is a decision tree $T^*$ satisfying
\begin{align*}
  \dist_\mathcal{D}(T^*,\Gamma)< \frac{1}{2N}\qquad \underset{\bx\sim \mathcal{D}}{\E}[\depth_{T^*}(\bx)]< \frac{\opt(\mathcal{S})}{2} \qquad
\Prx_{\bx\sim\mathcal{D}}\left[ T(\ParComplete_{\underline{j}}(\bz,\bx)) = \perp \right] \leq \frac{5\delta}{4} < \frac{1}{2}. 
\end{align*}

But this contradicts \Cref{claim:depth_error-new-abort}. \hfill{$\square$}

\subsection{Hardness amplification for {\sc DT-Estimation}}

Next, we amplify the distance given by~\Cref{lem:set_cover-reduction-with-aborts} using the following harndess amplification lemma: 

\begin{lemma} 
[Precise restatement of \Cref{lem:hardness-amplification}] 
\label{lem:hardness-amplification-actual} 
Let $f : \zo^n \to \zo$ and $\mathcal{D}$ be such that $f$ is $\eps$-far from every depth-$d$ $\delta$-abort decision tree where $\delta \ge 0.34$ under $\mathcal{D}$. Consider \[ f^{\oplus m}(x^{(1)},\ldots,x^{(m)}) \coloneqq f(x^{(1)})\oplus \cdots \oplus f(x^{(m)}),\]  the $m$-fold XOR composition of $f$ and $\mathcal{D}^m$ be the corresponding distribution over $(\zo^n)^m$. For any $\gamma > 0$, by taking $m = \Theta(\log(1/\gamma)/\eps)$, we get that $f^{\oplus m}$ is $(\frac1{2} - \gamma)$-far from every decision tree of depth $\Omega(dm)$ under $\mathcal{D}^m$. 
\end{lemma}

Note that the probability of error $\epsilon$ is taken over inputs that \emph{do not} abort. Thus, this statement is weaker than that of~\Cref{lem:hardness-amplification}. We need to allow the possibility of aborting in order to apply \cite{BKLS20}'s lemma. The proof of~\Cref{lem:hardness-amplification-actual} consists of two parts, each of which introduces another layer of XOR composition in order to amplify the error. First, we amplify the error from $\epsilon$ to a constant $O(1)$ and then from $O(1)$ to exponentially close to $\frac{1}{2}$. Each of these two steps uses an XOR lemma from~\cite{BKLS20}
and~\cite{Dru12} respectively. We now state these lemmas and then proceed with the proof of~\Cref{lem:hardness-amplification-actual}. 


\begin{lemma}[Lemma 1 of~\cite{BKLS20}]
\label{lem:brody} 
Let $f: \zo^n \to \zo$ and $\mathcal{D}$ be such that $f$ is $\eps$-far from every depth-$d$ $\delta$-abort tree where $\delta\ge 0.34$ under $\mathcal{D}$.  By taking $m = \Theta(1/\eps)$, we get that $f^{\oplus m}$ is $\frac1{800}$-far from every decision tree of depth $\Omega(dm)$ under $\mathcal{D}^m$. 
\end{lemma} 


\begin{lemma}[Theorem 1.3 of~\cite{Dru12}]
\label{lem:drucker}
Let $f: \zo^n \to \zo$ and $\mathcal{D}$ be such that $f$ is $\eps$-far from every depth-$d$ decision tree under $\mathcal{D}$.  For every $m \in \N$ and $\alpha \in [0,1]$, we get that $f^{\oplus m}$ is 
\[ \frac{1}{2} \big(1-(1-2\eps+6\alpha\ln(2/\alpha)\eps)^m\big)\]
far from every decision tree of depth $\alpha \eps d m$ under $\mathcal{D}^m$.
\end{lemma} 

Note that the original version of~\Cref{lem:drucker} in \cite{Dru12} holds for randomized decision trees rather than deterministic ones; however, the above version is equivalent. If $f$ is $\eps$-far from all depth-$d$ randomized decision trees, then clearly it is $\eps$-far from all deterministic ones since randomness can only add power. On the other hand, suppose $f$ is $\eps$-far from all depth-$d$ deterministic decision trees. Consider a depth-$d$ randomized decision tree $T(x,r)$ that in addition to $x$ takes in a random string~$r$.  The distance between $T$ and $f$ is given by
\[ \underset{\br}{\E} [\Prx_{\bx} [T(\bx,\br) \neq f(\bx)]]. \]
For each fixed $r$, we have that $\Pr[T(\bx,r) \neq f(\bx)]$ must be at least an $\eps$ fraction of total inputs. Thus, by linearity of expectation, the randomized decision tree must also be $\eps$-far from $f$. 


\begin{proof}[Proof of~\Cref{lem:hardness-amplification-actual}] 
Consider $f$ and $\mathcal{D}$ as in the statement of~\Cref{lem:hardness-amplification-actual}, and simply apply~\Cref{lem:brody} with $m_1 = \Theta(1/\eps)$. What results is a function $f^{\oplus m_1}$ that is $\frac{1}{800}$-far from every decision tree of depth $\Omega(d m_1)$ under $\mathcal{D}^{m_1}$.  Next, apply~\Cref{lem:drucker} to $f^{\oplus m_1}$ with $m_2 = \Theta(\log(1/\gamma))$, $\eps = \frac1{800}$, and by choosing $\alpha$ such that $6\alpha\ln(2/\alpha)=1$. Then, simplifying the expression in~\Cref{lem:drucker}, we get that $(f^{\oplus m_1})^{\oplus m_2}$ is 
\begin{align*}
    \frac{1}{2} \big(1-\left(1-\lfrac{1}{800}\right)^{m_2}\big) 
    &=\frac{1}{2} - \frac{1}{2}\left(\frac{799}{800}\right)^{\Theta(\log(1/\gamma))}\\
    &\ge \frac{1}{2}- 2^{-\log(1/\gamma)}\\
    &= \frac{1}{2} - \gamma
\end{align*}
far from every decision tree of depth $\Omega(d m_1 m_2)$ under $(\mathcal{D}^{m_1})^{m_2}$.  

Globally, we define $m = m_1 \cdot m_2 = \Theta(\log(1/\gamma)/\eps)$ so that $(f^{\oplus m_1})^{\oplus m_2}= f^{\oplus m}$. Then, $f^{\oplus m}$ is $\frac{1}{2}- \gamma$ far from decision trees of depth $\Omega(dm) = \Omega(d \log(1/\gamma)/\eps)$ under $\mathcal{D}^m$ as desired. 
\end{proof} 




\begin{figure}[h]
\begin{center}
\forestset{
 default preamble={
 for tree={
  circle,
  minimum size=0.7cm,
  inner sep=0.1pt,
  draw,
  align=center,
  anchor=north,
  fill=white,
  l sep=1cm,
  s sep=0.2cm,
 }
 }
}
\begin{forest}
[{},name=topgate, s sep=0.4cm,circle split,
     [$\lor$,name=lormain1,s sep=0.2cm, calign=center,
        [{},name=oplus11, s sep=0.4cm, inner sep=-0.0cm,circle split,
            [$y_{111}$,name=y111,tier=inputl,draw=white]
            [$y_{121}$,name=y1l1,tier=inputl,draw=white]
        ]
        [{},name=ldots11,draw=white,before computing xy={s=-30}]
        [$\ldots$,name=ldots21,draw=white,before computing xy={s=0}]
        [{},name=ldots31,draw=white,before computing xy={s=30}]
        [{},name=oplus21,s sep=0.4cm,inner sep=-0.0cm,circle split,
            [$y_{n11}$,name=yn11,tier=inputl,draw=white]
            [$y_{n21}$,name=ynl1,tier=inputl,draw=white]
        ]
     ]
     [{},name=ldotsmiddleleft,draw=white,before computing xy={s=-45}]
     [$\ldots$,name=ldotsmiddle,draw=white,before computing xy={s=0}]
     [{},name=ldotsmiddleright,draw=white,before computing xy={s=45}]
     [$\lor$,name=lormain2,s sep=0.2cm, calign=center,
        [{},name=oplus12, s sep=0.4cm, inner sep=-0.0cm,circle split,
            [$y_{11m}$,name=y112,tier=inputl,draw=white]
            [$y_{12 m}$,name=y1l2,tier=inputl,draw=white]
        ]
        [{},name=ldots12,draw=white,before computing xy={s=-30}]
        [$\ldots$,name=ldots22,draw=white,before computing xy={s=0}]
        [{},name=ldots32,draw=white,before computing xy={s=30}]
        [{},name=oplus22,s sep=0.4cm,inner sep=-0.0cm,circle split,
            [$y_{n1m}$,name=yn12,tier=inputl,draw=white]
            [$y_{n2 m}$,name=ynl2,tier=inputl,draw=white]
        ]
     ]
 ]
 \draw[black] (-1.75,-0.75) .. controls +(south east:1cm) and +(south west:1cm) .. (1.75,-0.75);
 \draw[black] (-4.25,-2.75) .. controls +(south east:0.5cm) and +(south west:0.5cm) .. (-2,-2.75);
 \node[xshift=-0.65cm,yshift=0.25cm] at (-2,-0.75) {fan-in $m$};
 \node[xshift=-0.75cm,yshift=0cm] at (-4.25,-2.5) {fan-in $n$};
 \node[draw,black,circle split,rotate=90,inner sep=-0.0cm,minimum size=0.7cm] at (topgate) {};
 \node[draw,black,circle split,rotate=90,inner sep=-0.0cm,minimum size=0.7cm] at (oplus21) {};
 \node[draw,black,circle split,rotate=90,inner sep=-0.0cm,minimum size=0.7cm] at (oplus11) {};
 \node[draw,black,circle,inner sep=-0.0cm,minimum size=0.7cm] at (lormain1) {};
  \node[draw,black,circle split,rotate=90,inner sep=-0.0cm,minimum size=0.7cm] at (oplus22) {};
 \node[draw,black,circle split,rotate=90,inner sep=-0.0cm,minimum size=0.7cm] at (oplus12) {};
 \node[draw,black,circle,inner sep=-0.0cm,minimum size=0.7cm] at (lormain2) {};
\end{forest}
\caption{A depth-$3$ circuit for $(\Gamma_{\oplus 2})^{\oplus m}$ consisting of one top gate that is a PARITY connected to $m$ independent copies of the circuit for $\Gamma_{\oplus 2}$.}
\label{fig:circuit-for-GammaLParitym}
\end{center}
\end{figure}
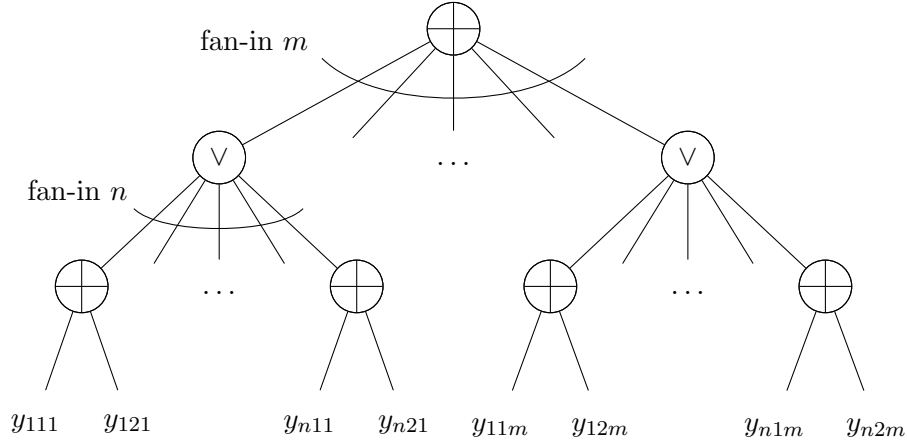

\subsection{Proof of~\Cref{thm:test}}

With~\Cref{lem:set_cover-reduction-with-aborts,lem:hardness-amplification-actual} in hand, we are now ready to prove~\Cref{thm:test}.  Given a size-$N$ instance $\mathcal{S}$ of $(k,k')$-\setcover with $n$ sets, we apply~\Cref{lem:Set-Cover-to-DTs,lem:set_cover-reduction-with-aborts} with $\ell = 2$ to obtain a $\poly(N)$-time reduction that produces a function $\Gamma_{\oplus 2} : \zo^{2n}\to \zo$ and the generator for a distribution $\mathcal{D}_{\oplus 2}$ over $\zo^{2n}$ satisfying:
\begin{enumerate} 
\item[$\circ$] If $\opt(\mathcal{S}) \le k$, then $\Gamma_{\oplus 2}$ is a $2k$-junta under $\mathcal{D}_{\oplus 2}$. 
\item[$\circ$] If $\opt(\mathcal{S}) > k'$, then for any $\delta < 0.4$, any $\delta$-abort decision tree of depth $k'/20$ is $\Omega(\frac1{N})$-far from $\Gamma_{\oplus 2}$ under $\mathcal{D}_{\oplus 2}$. 
\end{enumerate} 

Next, we consider $(\Gamma_{\oplus 2})^{\oplus m}$ and $(\mathcal{D}_{\oplus 2})^m$ where $m = \Theta(N^2)$:  
\begin{itemize}
    \item[$\circ$] If $\opt(\mathcal{S})\le k$, then $(\Gamma_{\oplus 2})^{\oplus m}$ is a $2km$-junta under $(\mathcal{D}_{\oplus 2})^m$.  Such a junta can be computed by a decision tree of depth $d\coloneqq 2km$. 
    \item[$\circ$] If $\opt(\mathcal{S}) > k'$, then by~\Cref{lem:hardness-amplification-actual}, $(\Gamma_{\oplus 2})^{\oplus m}$ is $(\frac1{2} - 2^{-N})$-far from decision trees of depth $d' \coloneqq \Omega(k'm)$ under $(\mathcal{D}_{\oplus 2})^m$.  
\end{itemize}

Note that the circuit representation for $(\Gamma_{\oplus 2})^{\oplus m}$ and generator for $(\mathcal{D}_{\oplus 2})^m$ can be constructed in $\poly(N)$ time from those for $\Gamma_{\oplus 2}$ and $\mathcal{D}_{\oplus 2}$ by simply feeding $m$ copies of $f$ into an XOR gate and by using $m$ copies of $\mathcal{D}$. See \Cref{fig:circuit-for-GammaLParitym} for an illustration of this circuit. Since $(\Gamma_{\oplus 2})^{\oplus m}$ is a function over $2nm \le O(N^3)$ variables and $d' \ge \Omega(d\log N)\ge \Omega(d\log d)$, \Cref{thm:test} now follows by applying~\Cref{thm:nonparameterized-SC} with $\beta \in (0,1)$ being any constant.

\section{Proof of \Cref{thm:proper}}

The PAC learning lower bound from \Cref{subsub:PAC-learning-hardness} applies to properly learning decision trees. In this setting, the concept class is $\mathcal{T}=\{T:\zo^n\to\zo\mid T\text{ is a decision tree}\}$. So the learner is allowed to output a decision tree hypothesis that may be much larger than the target. One could instead consider the problem of properly learning the class of size-$s$ decision trees: $\mathcal{T}_s=\{T:\zo^n\to\zo\mid T\text{ is a size-}s\text{ decision tree}\}$. This problem is strictly harder than learning decision trees since the output must satisfy a size constraint. Indeed, against this class, we are able to adapt the proof of \Cref{thm:DNF-construction-hardness} to obtain a stronger lower bound.
\begin{sloppypar}
\begin{theorem}
Assuming randomized ETH, there is a constant $\lambda\in (0,1)$ such that {\sc DT-Construction}$(s,\frac{1}{n})$ cannot be solved in time $n^{\lambda\log s}$ if the algorithm has to return a size-$s$ DNF hypothesis, even when the function is promised to be a $\log s$-junta.
\end{theorem}
\end{sloppypar}

\begin{proof}
This proof is a combination of the proofs of \Cref{thm:optimal-hardness-of-dt-construction-under-better-sc-lbs,thm:DNF-construction-hardness}. The analysis is similar so we only outline the important details here. In particular, let $\mathcal{S}=(S,U,E)$ be an $N$-vertex $\left(k,\frac{1}{2}\left(\frac{\log N}{\log\log N}\right)^{1/k}\right)$-\setcover instance where $k$ is taken to be $k=\frac{1}{2}\cdot\frac{\log\log N}{\log\log\log N}$ for $N$ large enough so that $32k<\frac{1}{2}\left(\frac{\log N}{\log\log N}\right)^{1/k}$. Using \Cref{thm:SC_hardness}, there is a constant $c\in (0,1)$ such that $\mathcal{S}$ cannot be solved in time $N^{ck}$. We derive a contradiction for any algorithm for {\sc DT-Construction}$(s,\frac{1}{n})$ that returns a size-$s$ DNF and runs in time $n^{\lambda \log s}$ for $\lambda\le c/5$. 

Use \Cref{lem:NGamma-has-no-small-DNF} with $\ell=4$ (as in \Cref{thm:optimal-hardness-of-dt-construction-under-better-sc-lbs}) to obtain the target function $\overline{\Gamma}_{\oplus 4}:\zo^{4N}\to\zo$ and the distribution $\mathcal{D}_{\oplus 4}$. Run {\sc DT-Construction}$(s,\frac{1}{4N})$ with $s\coloneqq 2^{4k}$ on $\overline{\Gamma}_{\oplus 4}:\zo^{4N}$ and $\mathcal{D}_{\oplus 4}$ for $N^{5\lambda k}$ time steps where $\lambda\le c/5$. Output yes if and only if the DNF formula returned by the algorithm as size at most $2^{4k}$ and error less than $1/(4N)$. The correctness of the no case follows from the fact that $32k<\frac{1}{2}\left(\frac{\log N}{\log\log N}\right)^{1/k}$ and so the DNF lower bound from \Cref{lem:NGamma-has-no-small-DNF} ensures $2^{\opt(\mathcal{S})\ell/16}>2^{4k}$. Since the algorithm returns a size-$2^{4k}$ DNF formula if one exists, this separation between the DNF sizes is sufficient to establish correctness.
\end{proof}

As in the case of \Cref{cor:pac-learning-hardness-of-dts-dnfs}, this theorem yields hardness of properly PAC learning the class of size-s decision trees.

\begin{corollary}
Assuming randomized ETH, there is a constant $\lambda \in (0,1)$ such that the class $\mathcal{T}_s$ of size-s decision trees cannot be distribution-free, properly PAC learned to accuracy $\eps=1/n$ in time $n^{\lambda\log s}$. The same result also holds for properly learning size-$s$ DNFs and CNFs.
\end{corollary}

\section{Proof of \Cref{thm:opt}}
\begin{sloppypar}
In this section, we outline a concrete path towards proving optimal lower bounds for {\sc DT-Construction}. In particular, we show that better lower bounds for gapped set cover yields better lower bounds for {\sc DT-Construction}. Specifically, the main theorem assumes \Cref{conj:optimal-SetCover-hardness} and proves an $n^{\Omega(\log s)}$ lower bound for {\sc DT-Construction}$(s,\frac{1}{n})$. 

\begin{theorem}
\label{thm:optimal-hardness-of-dt-construction-under-better-sc-lbs}
Assume \Cref{conj:optimal-SetCover-hardness}, then there is a constant $\lambda\in (0,1)$ such that \textsc{DT-Construction}$(s,\frac{1}{n})$ cannot be solved in time $n^{\lambda \log s}$, even when the target is a $\log s$-junta and the algorithm is allowed to return a DNF hypothesis.
\end{theorem}
\end{sloppypar}

\begin{proof}
Let $\beta\in (0,1)$ be as in the statement of the \Cref{conj:optimal-SetCover-hardness}. Assume there is an algorithm for {\sc DT-Construction}$(s,\frac{1}{n})$ running in time $n^{\lambda\cdot\log s}$ for any $\lambda\le (1-\beta)/(40\log e)$. Then, following the proof strategy of \Cref{thm:DNF-construction-hardness}, we derive a contradiction by solving $(k,k\cdot(1-\beta)\ln N)$-{\sc Set-Cover} over $N$ vertices in time $N^{5\lambda k}$. Let $\mathcal{S}=(S,U,E)$ be an $N$-vertex $(k,k\cdot(1-\beta)\ln N)$-{\sc Set-Cover} instance for $k\in \N$. Using \Cref{lem:given-set-cover-construct-NGamma-which-is-hard-to-approximate-by-DNFs} with $\ell=4$, we obtain the target function $\overline{\Gamma}_{\oplus 4}:\zo^{4N}\to\zo$ and the distribution $\mathcal{D}_{\oplus 4}$. We run the algorithm for {\sc DT-Construction}$(s,\frac{1}{4N})$ on $\overline{\Gamma}_{\oplus 4}$ and $\mathcal{D}_{\oplus 4}$ with $s\coloneqq 2^{4k}$ and terminate it after $N^{5\lambda k}$ time steps. The output is some DNF formula $F$. We estimate the error of $F$ over $\mathcal{D}_{\oplus 4}$ and output \textsc{Yes} if it's less than $1/(4N)$ and \textsc{No} otherwise.

\paragraph{Runtime.}{By \Cref{lem:given-set-cover-construct-NGamma-which-is-hard-to-approximate-by-DNFs}, we can construct the circuit for $\overline{\Gamma}_{\oplus 4}:\zo^{4N}\to\zo$ and generator for $\mathcal{D}_{\oplus 4}$ in $\poly(N)$-time. Moreover, we can use random sampling to efficiently estimate the error of $F$ over $\mathcal{D}_{\oplus 4}$.  Therefore, the runtime of the reduction is dominated by $N^{5\lambda k}$. }

\paragraph{Correctness.}{ We handle the yes case and the no case separately.
\subparagraph{Yes case: $\opt(\mathcal{S})\le k$.}{
By \Cref{lem:given-set-cover-construct-NGamma-which-is-hard-to-approximate-by-DNFs}, $\overline{\Gamma}_{\oplus 4}$ is a $4k$-junta over $\mathcal{D}_{\oplus 4}$. Therefore, it is a decision tree of size $s=2^{4k}$ and {\sc DT-Construction}$(s,\frac{1}{4N})$ runs in time $$(4N)^{\lambda\cdot\log s}=(2N)^{4\lambda k}\le N^{5\lambda k}.$$ The output is DNF formula with error at most $1/(4N)$. It follows that our algorithm outputs \textsc{Yes} with high probability.
}
\subparagraph{No case: $\opt(\mathcal{S})> k\cdot(1-\beta) \ln N$.}{
By \Cref{lem:given-set-cover-construct-NGamma-which-is-hard-to-approximate-by-DNFs}, any DNF for $\overline{\Gamma}_{\oplus 4}$ with size at most $2^{\opt(\mathcal{S})/8}$ has error at least $1/(4N)$. Using the assumption on $\opt(\mathcal{S})$:
\begin{align*}
    2^{\opt(\mathcal{S})/8}&> N^{k(1-\beta)/(8\log e)}\\
    &\ge N^{5\lambda k} \tag{$(1-\beta)/(40\log e)\ge \lambda $}
\end{align*}
}
which shows that the DNF output by the algorithm must have error at least $1/(4N)$. It follows that our algorithm outputs \textsc{No} with high probability.}
\end{proof}

As discussed in \Cref{subsub:PAC-learning-hardness}, this lower bound for {\sc DT-Construction} implies a lower bound for PAC learning decision trees.

\begin{corollary}
Assume \Cref{conj:optimal-SetCover-hardness}, then there is a constant $\lambda \in (0,1)$ such that decision trees cannot be distribution-free, properly PAC learned to accuracy $\eps=1/n$ in time $n^{\lambda\log s}$ where $s$ is the size of the decision tree target. The same result also holds for properly learning DNFs and CNFs.
\end{corollary}

The proof of this corollary is identical to that of \Cref{cor:pac-learning-hardness-of-dts-dnfs}.

\section*{Acknowledgments}

We thank the SODA reviewers for their useful comments and feedback. 

Caleb, Carmen, and Li-Yang are supported by NSF awards 1942123, 2211237, and 2224246. Caleb is also supported by an NDSEG fellowship. 

\bibliography{ref}
\bibliographystyle{alpha}

\appendix
\section{Hardness of Approximating Set Cover}

\label{appendix:set_cover_hardness}
We first state a lemma due to \cite{Lin19}, translated into our notation.

\begin{lemma}[Lin's lemma {\cite[Lemma 3.6]{Lin19}}]
\label{lem:lins-lemma-3.6}
There is an algorithm which given $k\in\N$, $\delta>0$ with $(1+1/k^3)^{1/k}\le (1+\delta)/(1+\delta/2)$ and $(1+\delta/2)^k\ge 2k^4$ and a SAT instance $\phi$ with $n$ variables and $Cn$ clauses, where $n$ is much larger than $k$ and $C$, outputs an integer $N\le 2^{n/k}+n/k^3$ and a set cover instance $\mathcal{S}=(S,U,E)$ satisfying
\begin{itemize}
    \item $|S|+|U|\le N$;
    \item if $\phi$ is satisfiable, then $\opt(\mathcal{S})\le k$;
    \item if $\phi$ is unsatisfiable, then $\opt(\mathcal{S})>\frac{1}{1+\delta}\paren{\frac{\log N}{\log\log N}}^{1/k}$
\end{itemize}
\end{lemma}
\noindent
The exact version of this lemma we use is the following.
\begin{lemma}[Reducing SAT to \textsc{Set-Cover}]
\label{lem:SAT_to_SC}
There is an algorithm that takes an $n$-variate SAT instance $\varphi$ of size $|\varphi|$ and an integer $k\ge 100$ with $k^2\le n/\log n$ and produces a set cover instance $\mathcal{S}$ of size $N\le 2^{2|\varphi|/k}$ in time $\le 2^{5|\varphi|/k}$ such that
\begin{enumerate}
    \item if $\varphi$ is satisfiable then $\opt(S)\le k$;
    \item if $\varphi$ is unsatisfiable then $\opt(S)>\frac{1}{2}\left(\frac{\lg N}{\lg\lg N}\right)^{1/k}$.
\end{enumerate}
\end{lemma}

\begin{proof}
We use \Cref{lem:lins-lemma-3.6} with $\delta=1/2$. For this value of $\delta$, if $k\ge 100$, then both conditions $(1+1/k^3)^{1/k}\le (1+\delta)/(1+\delta/2)$ and $(1+\delta/2)^k\ge 2k^4$ of \Cref{lem:lins-lemma-3.6} are satisfied. Moreover, an inspection of the proof of \Cref{lem:lins-lemma-3.6} shows that the condition ``$n$ is much larger than $k$'' in the lemma statement means $k^2\le n/\log n$.

Therefore, \Cref{lem:lins-lemma-3.6} returns a set cover instance $\mathcal{S}$ satisfying
\begin{enumerate}
    \item if $\varphi$ is satisfiable then $\opt(S)\le k$;
    \item if $\varphi$ is unsatisfiable then $\opt(S)>\frac{1}{1+\delta}\left(\frac{\lg N}{\lg\lg N}\right)^{1/k}$.
\end{enumerate}
By our choice of $\delta$, $1/(1+\delta)=2/3>1/2$ as desired. Since $n\le |\varphi|$, the size of the set cover instance is $N\le 2^{|\varphi|/k+|\varphi|/k^3}\le 2^{2|\varphi|/k}$. The runtime of the reduction is $\le 2^{5|\varphi|/k}$.
\end{proof}

We can now prove the main theorem from \Cref{subsection:results-on-set-cover}.

\begin{proof}[Proof of \Cref{thm:SC_hardness}]
Suppose there exists an algorithm that can solve $\left(k,\frac{1}{2}\left(\frac{\log N}{\log\log N}\right)^{1/k}\right)$-\setcover on $N$ vertices with high probability in time $N^{ck}$. Then we show how to solve SAT with high probability for SAT formulas with $n$ variables in time $2^{3cn}$.  

Let $\varphi$ be a SAT instance with $n$ variables. Choose $k\in \N$ so that
$$
k=\frac{1}{2}\cdot\frac{\log\log(2^{2n/k})}{\log\log\log(2^{2n/k})}=\frac{1}{2}\cdot\frac{\log(2n/k)}{\log\log(2n/k)}.
$$
Given $n$ this equation can be numerically solved efficiently, and $k$ will be some value between $\log\log n$ and $\log n$. We then apply \Cref{lem:SAT_to_SC} with this value of $k$ to obtain a set cover instance $\mathcal{S}$ of size $N\le 2^{2n/k}$ in time $\le 2^{5n/k}$. If $N<2^{2n/k}$ then we add dummy items/dummy sets to the universe so that $N=2^{2n/k}$. Note this padding will not affect the optimal set cover for the optimal set cover size. Hence, by construction, we have an instance of $\left(k,\frac{1}{2}\left(\frac{\log N}{\log\log N}\right)^{1/k}\right)$-\setcover of size $N$ where
$$
k=\frac{1}{2}\cdot\frac{\log\log(N)}{\log\log\log(N)}.
$$
We can therefore run our algorithm for set cover on this instance $\mathcal{S}$ and output ``\textsc{Yes}'' if the algorithm outputs \textsc{Yes} and ``\textsc{No}'' if the algorithm outputs \textsc{No}.

\paragraph{Runtime.}{
Our reduction runs in time
\begin{align*}
    2^{5n/k}+N^{ck}&\le 2^{5n/k}+(2^{2n/k})^{ck}\\
    &=2^{5n/k}+2^{2cn}\\
    &\le 2^{3cn}.
\end{align*}
}
\begin{sloppypar}
\paragraph{Correctness.}{
By assumption, the set cover algorithm solves $\left(k,\frac{1}{2}\left(\frac{\log N}{\log\log N}\right)^{1/k}\right)$-\setcover with high probability and therefore by \Cref{lem:SAT_to_SC} our algorithm solves SAT with high probability. We also note that $2k^{k}<\frac{\log N}{\log\log N}$ and so $k<\frac{1}{2}\left(\frac{\log N}{\log\log N}\right)^{1/k}$ for our choice of $k$. If instead, one were to choose e.g. $k=\log\log n$, then $2k^k>\frac{\log N}{\log \log N}$ and so the set cover instance would fail to determine the satisfiability of $\varphi$. 
}
\end{sloppypar}

It follows that if SAT cannot be solved in randomized time $O(2^{\delta n})$ for some $\delta\in (0,1)$ then $\left(k,\frac{1}{2}\left(\frac{\log N}{\log\log N}\right)^{1/k}\right)$-\setcover cannot be solved in randomized time $N^{\delta/3 \cdot k}$.
\end{proof}

\section{Proof of \Cref{prop:Dj_equals_D}}
\label{appendix:Dj-equals-D}
We first compute: 
\begin{align*}
    \pryd{\by=y}&=\pryd{\BlockwisePar(\by) = \BlockwisePar(y)}\\
    & \ \ \ \cdot\pryd{\by=y\mid \BlockwisePar(\by) = \BlockwisePar(y)}\tag{Law of total probability}\\
    &=\prxd{\bx=\BlockwisePar(y)}\cdot \underset{\by\sim \mathcal{U}_{n\cdot\ell}}{\Pr}\left[{\by=y\mid \BlockwisePar(\by) = \BlockwisePar(y)}\right]\tag{Definition of $\mathcal{D}_{\oplus \ell}$}\\
    &=\prxd{\bx=\BlockwisePar(y)}\cdot \prod_{i\in[n]}\underset{\by_i\sim \mathcal{U}_{\ell}}{\Pr}\left[{\by_i=y_i\mid \oplus\by_i = \oplus y_i}\right]\tag{Independence of $\by_i$'s}\\
    &=\prxd{\bx=\BlockwisePar(y)}\cdot \prod_{i\in[n]}2^{-(\ell-1)}
\end{align*}
where the last step follows from the fact that conditioning on the parity of $\by_i$ being a specific bit removes $1$ out of $\ell$ degrees of freedom.  With an analogous calculation for $\mathcal{D}^j_{\oplus\ell}$, we obtain
\begin{align*}
    \prydj{\by=y}&=\prydj{\BlockwisePar(\by) = \BlockwisePar(y)} \\ 
    & \ \ \ \cdot\prydj{\by=y\mid \BlockwisePar(\by) = \BlockwisePar(y)}\tag{Law of total probability}\\
    &=\prxd{\bx=\BlockwisePar(y)}\cdot\prod_{i\in[n]}\underset{\by_i\sim\mathcal{U}_{\ell-1}}{\Pr}\left[\by_i=y_i^{-j}\right]\tag{Definition of $\mathcal{D}_{\oplus \ell}^j$}\\
    &=\prxd{\bx=\BlockwisePar(y)}\cdot \prod_{i\in[n]}2^{-(\ell-1)}
\end{align*}
where $y_i^{-j}\in\zo^{\ell-1}$ is the string $y_i$ with its $j$th bit removed.

\section{PAC learning}
\label{appendix:pac-learning}

In the realizable PAC learning model \cite{Val84}, there is an unknown distribution $\mathcal{D}$ and some unknown \textit{target} function $f\in\mathcal{H}$ from a fixed \textit{concept} class $\mathcal{H}$ of functions over a fixed domain. An algorithm for learning $\mathcal{H}$ over $\mathcal{D}$ takes as input $\eps\in (0,1)$ and has oracle access to an \textit{example oracle} $\textnormal{EX}(f,\mathcal{D})$. The algorithm can query the example oracle to receive a pair $(x,f(x))$ where $x\sim\mathcal{D}$ is drawn independently at random. The goal is to output a \textit{hypothesis} $h$ such that $\dist_{\mathcal{D}}(f,h)\le \eps$. Since the example oracle is inherently randomized, any learning algorithm is necessarily randomized. So we require the algorithm succeed with some fixed probability e.g. $2/3$. A learning algorithm is \textit{proper} if it always outputs a hypothesis $h\in\mathcal{H}$. 

Formally, we use the following definition for PAC learning decision trees. 
\begin{definition}[PAC learning decision trees]
Let $\mathcal{T}=\{T:\zo^n\to\zo\mid T\text{ is a decision tree}\}$ be the class of decision trees over a fixed domain $\zo^n$. A distribution-free learning algorithm $\mathcal{L}$ learns $\mathcal{T}$ in time $t(n,s,\eps)$ if for all distributions $\mathcal{D}$ and for all $T\in\mathcal{T},\eps\in(0,1)$, $\mathcal{L}$ with oracle access to $\textnormal{EX}(T,\mathcal{D})$ runs in time $t(n,|T|,\eps)$ and with probability $2/3$ outputs $h:\zo^n\to\zo$ such that $\dist_{\mathcal{D}}(T,h)\le \eps$. Furthermore, $\mathcal{L}$ is proper if $h\in \mathcal{T}$. 
\end{definition}


\end{document}